\documentclass[11pt,a4paper]{article}
\usepackage[a4paper,margin=1.04in,includefoot]{geometry}
\usepackage[T1]{fontenc}
\usepackage{lmodern}
\usepackage{iftex}
\ifPDFTeX
  \input{glyphtounicode}
\fi
\usepackage{amsmath,amssymb,amsfonts,amsthm,mathtools}
\usepackage{xcolor,booktabs,array,longtable,tabularx}
\usepackage[numbers,sort&compress]{natbib}
\usepackage[title]{appendix}
\usepackage{etoolbox}
\usepackage[hidelinks]{hyperref}
\theoremstyle{plain}
\newtheorem{theorem}{Theorem}
\newtheorem{proposition}[theorem]{Proposition}
\newtheorem{lemma}[theorem]{Lemma}
\newtheorem{corollary}[theorem]{Corollary}
\theoremstyle{definition}
\newtheorem{definition}[theorem]{Definition}
\theoremstyle{remark}
\newtheorem{remark}[theorem]{Remark}
\AtBeginEnvironment{appendices}{
  \setcounter{figure}{0}
  \setcounter{table}{0}
  \setcounter{equation}{0}

}

\newcommand{\F}{\mathbb F}
\newcommand{\cC}{\mathcal C}
\newcommand{\cG}{\mathcal G}
\newcommand{\cV}{\mathcal V}
\newcommand{\cB}{\mathcal B}
\newcommand{\cA}{\mathcal A}

\newcommand{\rank}{\operatorname{rank}}
\newcommand{\supp}{\operatorname{supp}}
\newcommand{\wtN}{\operatorname{wt}_{\rm N}}
\newcommand{\wtCl}{\operatorname{wt}_{\rm cl}}
\newcommand{\St}{\operatorname{St}}
\newcommand{\Stedge}{\operatorname{St}_{\rm edge}}
\newcommand{\Cl}{\operatorname{Cl}}

\hypersetup{
 hidelinks,
 pdflang={en-US},
 pdftitle={Binary Multiple-Node-Erasure-Correcting Codes over Complete Graphs:
 Constructions, q-Ary Metric Balls, and Duality},
 pdfauthor={Aryeh Lev Zabokritskiy (Yohananov)}
}
\title{Binary Multiple-Node-Erasure-Correcting Codes over Complete Graphs:
Constructions, \(q\)-Ary Metric Balls, and Duality}
\author{Aryeh~Lev Zabokritskiy (Yohananov)\\[0.45em]
\small Department of Computer Science,\\
\small Tel-Hai University of Kiryat Shmona in the Galilee,\\
\small Kiryat Shmona, Israel\\[0.3em]
\small MIGAL--Galilee Research Institute, Kiryat Shmona, Israel\\[0.45em]
\small \href{mailto:yuhanalev@telhai.ac.il}{yuhanalev@telhai.ac.il}}
\date{}
\begin{document}
\maketitle
\begin{abstract}
Codes over graphs store symbols on edges of a complete undirected graph on
\(n\) vertices with a self-loop at each vertex; erasing a node removes all its
incident edges. Extending Yohananov, Efron, and Yaakobi's construction, we obtain
binary codes correcting any three node erasures at infinitely many prime lengths
\(n\), unconditionally, with redundancy \(3n-2\), one bit above the graph
Singleton bound. The earlier construction has the same redundancy but requires
\(2\) to be primitive modulo \(n\), which is not known for infinitely many primes.
We give a deterministic decoder for these triple-node codes using \(O(n^2)\)
binary operations and construct Singleton-optimal binary triple-node codes at
\(n=6,8,10,12\).
When \(2\) is primitive modulo an odd prime \(n\), we extend their double- and
triple-node codes to binary codes correcting \(\rho\) node erasures with redundancy \(\rho n-(\rho-1)\)
for \(2\leq\rho\leq(n+1)/2\), saving \(\rho-1\) bits over the binary symmetric
rank-metric construction when \(\rho<(n+1)/2\).
Over every finite field of size \(q\), the node weight of a word is the least
number of vertices whose erasure removes its support. We give exact formulas for
ball volumes in this node metric and prove that, for fixed \(q,t\), the
radius-\(t\) volume is asymptotic to \(\binom nt q^{tn-\binom t2}\) as
\(n\to\infty\). We obtain existence and covering bounds; for fixed \(q,\rho\),
we show that linear \(\rho\)-node-erasure-correcting codes exist with redundancy
at most \(\rho\log_q n+O_{q,\rho}(1)\) above the graph Singleton bound.
Finally, reformulating Yohananov and Yaakobi's criterion, we show that a linear
code over any finite field is Singleton-optimal for \(\rho\) node erasures if
and only if its dual is Singleton-optimal for erasures of the coordinates inside
any \(n-\rho\) vertices.
\end{abstract}


\section{Introduction}

Codes over graphs of length \(n\) store symbols on the edges of a complete
undirected graph on \(n\) vertices equipped with one self-loop at each vertex, rather than on a linearly ordered sequence.
In the node-erasure model, failure of a vertex erases every edge incident with
it, including its self-loop.  The resulting node distance between two graph words
is the minimum number of vertices whose removal makes the two words
identical, equivalently the vertex-cover number of the support
of their difference.  This model was introduced in
\cite{YohananovYaakobi2019}, which also constructed optimal binary codes for
two failures at every prime length \(n\geq5\); a simpler optimal double-node construction
at the same lengths and a nearly optimal construction for three failures were developed in
\cite{YohananovEfronYaakobi2020}.
Related graph-based coding models include codes over trees under the tree
distance \cite{YohananovYaakobi2021Trees} and storage codes on coset graphs
\cite{BargSchwartzYohananov2024}.
In extremal combinatorics, the term graph-code refers to a related but
different notion, whose systematic study began in
\cite{AlonGujgiczerKornerMilojevicSimonyi2023}. A graph-code is a family of graphs on \(n\)
labeled vertices, viewed as binary vectors indexed by the edges of the
loopless complete graph, in which no two members differ by a graph from a
prescribed family closed under isomorphism
\cite{Alon2024GraphCodes,Alon2025GraphCodesSurvey}.
Binary codes without loop coordinates that correct \(\rho\) node erasures
are exactly the graph-codes for the family of nonempty graphs with
vertex-cover number at most \(\rho\).
Over fields of size \(q\ge n-1\), the original work on graph erasures
\cite{YohananovYaakobi2019} also gives Singleton-optimal
\(\rho\)-node-erasure codes for \(1\leq\rho<n\); their redundancy
\(\rho n-\binom\rho2\), the number of coordinates erased by \(\rho\) failed
vertices, is the least that the graph Singleton bound allows.  The binary
case is the harder setting addressed here.

All three parts of the paper study the coordinates erased by a fixed set of
failed vertices.  In the
construction part, the columns of a parity-check matrix indexed by these
coordinates must be linearly independent.  In the metric part, the words supported on such coordinate
sets form a family of subspaces whose union is a node-metric ball.  The complementary
coordinate sets, which form cliques on the surviving vertices, then lead to a duality between
optimal node-erasure codes and optimal clique-erasure codes.
The node-erasure-correcting constructions are binary; the ambient node-weight enumerators, ball volumes, general
coding bounds, and the duality hold over every finite field.

The node metric has recently reappeared in the study of asymptotically good
error-correcting codes over graphs.  Kopparty, Potukuchi, and Sha
\cite{KoppartyPotukuchiSha2025} showed that random linear codes attain the
optimal rate-distance tradeoff
\(R=(1-\delta)^2-o(1)\), where \(R\)
is the code rate and \(\delta\) is the relative node distance, and developed
symmetric tensor, concatenated, and dual-BCH constructions.  Their
binary model omits self-loops and their principal regime has
distance growing linearly with \(n\).  The present paper instead focuses on
low redundancy for a fixed number of failures, retains the loop
coordinates, and studies the exact combinatorial structure of the metric balls.

Chen, Cheraghchi, and Shagrithaya \cite{ChenCheraghchiShagrithaya2025}
have subsequently developed strongly explicit asymptotic erasure codes on
non-bipartite graphs, equivalently symmetric matrices with zero diagonal.  Their
concatenated construction attains rate at least
\((1-\sqrt\delta)^4-o(1)\) for sufficiently large block length, with encoding and erasure
decoding in quasi-linear time after a one-time preprocessing step.  Both recent
directions concern asymptotic regimes and leave the exact
redundancy questions for a fixed number of failures largely open.

There is also a precise connection with the cover metric for crisscross
array errors \cite{Roth1991}; node failures were already viewed as symmetric
crisscross erasures in \cite{YohananovYaakobi2019,YohananovEfronYaakobi2020}.  In a general matrix that metric permits separate row and
column covers.  A node failure in a symmetric matrix couples row \(i\) to
column \(i\) and counts the pair once, so the present metric is a
cover metric on symmetric matrices with coupled rows and columns.  Recent work on codes in the multi-cover
metric develops Singleton bounds, duality, and constructions for lists of
matrices \cite{MartinezPenas2026}.  There, duality is studied within the
multi-cover metric, where the dual of a linear code meeting the Singleton
bound need not meet it.  Here, the coordinates surviving a set of failed
vertices form the clique on the remaining vertices, so duality relates node
erasures to clique erasures.

For \(\rho=3\), the construction of
\cite{YohananovEfronYaakobi2020} has redundancy \(3n-2\), whereas the graph
Singleton bound is \(3n-3\).  Schmidt's binary symmetric rank-metric
construction \cite{Schmidt2010} (see also \cite[Thm.~5.3]{Schmidt2020}) was used for node erasures in
\cite[Constr.~5, Thm.~11]{YohananovYaakobi2019}; it has redundancy \(3n\) at odd
lengths \(n\geq7\), the least possible for binary symmetric rank-metric codes of minimum
rank distance seven \cite[Cor.~7]{Schmidt2010}.
Thus closing this one-bit gap
is specific to graphs: the rank-metric construction excludes every
nonzero codeword difference of rank at most six, while a triple-node code
only has to exclude the subfamily supported on the coordinates incident
with three vertices.

The published triple-node construction combines neighborhood and
slope-one diagonal checks with parity checks on the cyclic diagonals of the
fixed slope \(2\) in the symmetric labeling matrix, and it assumes that
\(2\) is primitive modulo the prime length.  Infinitude of
the prime lengths modulo which \(2\) is primitive would follow from
Artin's conjecture for the integer \(2\).  Keeping the neighborhood and slope-one diagonal checks of the published construction, we separate a
determinant condition, imposed at each failed triple, from the primitive-root
condition and allow the slope
to be chosen from \(\{2,-3,-15\}\).  This also gives codes at concrete prime lengths modulo
which \(2\) is not primitive: for example, \(2\) has order \(8\)
modulo \(17\), whereas \(-3\) has order \(16\).  We prove that, at every prime length
\(n>5\), each of these three slopes satisfies the determinant condition at all
failed triples, so the construction yields a triple-node code at such a length as soon as one of them is
primitive modulo \(n\). Applying Heath-Brown's primitive-root theorem for three integers
\cite{HeathBrown1986}, we obtain an unconditional infinite set of such prime
lengths.

We show that the resulting codes
have redundancy \(3n-2\).  On the other hand, we construct explicit binary codes at
\(n=6,8,10,12\) that attain the graph Singleton redundancy \(3n-3\), and for \(n=5\)
such a code follows from \cite[Thm.~18]{YohananovYaakobi2019}.  Thus the extra bit
is not inherent in the model; what remains open is an infinite optimal family.

\subsection{Contributions}

\begin{enumerate}
\item \emph{Triple-node cyclic constructions.}
For prime \(n\geq5\) and a primitive cyclic slope \(\lambda\), we characterize
the recovery of all failed triples by an explicit determinant criterion: no three of the \(n\) points \((1,\zeta^j,\zeta^{\lambda j})\), with \(\zeta\) a primitive \(n\)-th root of unity over \(\F_2\), are collinear. We also show that this condition fails at \(n=7\) for both primitive slopes, \(3\) and \(5\). By allowing the slope to vary,
we obtain binary triple-node codes of redundancy \(3n-2\) unconditionally at infinitely many
prime lengths. More generally, for every nonzero slope other than one, we compute
for each failed triple the exact dimension of the subcode supported on its
erased coordinates. This formula shows that if the
slope is not primitive, no triple is uniquely recoverable. For slopes
whose codes recover every triple we give a deterministic decoder using
\(O(n^2)\) binary operations, including syndrome computation.

\item \emph{Multi-slope codes.}
We prove that, when \(2\) is primitive modulo an odd prime \(n\), our binary multi-slope
codes, which keep the neighborhood and slope-one diagonal checks and use the
cyclic slopes \(2^k\) with \(1\leq k\leq\rho-2\),
correct every \(\rho\)-node erasure with exact redundancy
\(\rho n-(\rho-1)\) for \(2\leq\rho\leq(n+1)/2\).  This redundancy is \(\binom{\rho-1}2\)
above the graph Singleton bound and, for \(\rho<(n+1)/2\), \(\rho-1\) below the redundancy \(\rho n\)
of Schmidt's rank-metric construction.  At \(\rho=2\) and \(\rho=3\) these codes are the published double- and triple-node codes of \cite{YohananovEfronYaakobi2020}.

\item \emph{Optimal triple-node codes and loop completion.}
We construct explicit Singleton-optimal binary triple-node codes at
\(n=6,8,10,12\) by adding loop columns to an ordinary-edge check matrix.  For
every \(2\leq\rho\leq n\) we define such a matrix with \(\rho(n-1)\) rows and
prove that its erased-edge columns are independent for every
set of \(\rho\) failed vertices; for \(\rho=2\) and \(n\geq3\) this independence is known from a
dimensional dual hyperoval \cite{Yoshiara2008}. For \(\rho\geq3\),
suppose the check space \(\F_2^{\rho(n-1)}\) has a quotient of dimension
\(\rho n-\binom\rho2\) in which the erased-edge columns remain
independent for every failed set. Fix such a quotient and assign one
loop column to each vertex. We prove that the resulting code is Singleton-optimal
exactly when, for every failed set, its \(\rho\) loop columns complete
the corresponding erased-edge columns to a basis of the quotient.
At \(\rho=3\) the original check space already has the required dimension;
for \(\rho\geq4\) we exhibit no such quotient.

\item \emph{Node-metric balls and coding bounds.}
For arbitrary prime powers \(q\), we reduce the node-weight distribution exactly to the
distribution of vertex-cover numbers of simple graphs with nonzero edge labels
in \(\F_q\) and give an exact
inclusion--exclusion formula for ball volumes. We also prove that, for fixed \(q\) and \(t\), every ball of radius \(t\) has volume
\[
 \binom nt\,q^{\,tn-\binom t2}
 \left(1+O_{q,t}(nq^{-n})\right).
\]
We derive existence and covering bounds from these volumes. In particular, we show that, for fixed
\(\rho\), linear codes correcting \(\rho\) node erasures exist with redundancy at most
\(\rho\log_q n+O_{q,\rho}(1)\) above the graph Singleton bound.
At the other end of the radius range, we prove that the minimum number of centers
at radius \(n-m\), \(2\leq m\leq n\), is at least \(q^2\), with equality exactly when
\(n>q(R_q(m)-1)\), where \(R_q(m)\) is the diagonal \(q\)-color Ramsey
number for a clique on \(m\) vertices.

\item \emph{Clique enumeration and duality.}
Over every finite field, we give in closed form the ambient enumerator and
ball volumes of the clique weight, the number of vertices incident with the
nonzero coordinates of a word; the volumes recover the count in
\cite[Thm.~21]{YohananovYaakobi2019}.  Reformulating the full-rank
criterion of \cite[Lem.~17]{YohananovYaakobi2019}, we show that a code is
Singleton-optimal for \(\rho\) node erasures exactly when its dual is
Singleton-optimal for erasures of the coordinates inside any \(n-\rho\)
vertices.
\end{enumerate}

The arithmetic hypotheses of the two cyclic results are different.  For
the fixed slope \(2\), and hence for the general multi-slope construction,
the infinitude of prime lengths modulo which \(2\) is primitive is still predicted by Artin's
conjecture but is not known unconditionally.  The unconditional triple-node
argument does not identify which of the three slopes
is primitive for infinitely many primes.

The remainder of the paper is organized as follows.
Section~\ref{sec:preliminaries} fixes the looped graph model, the node metric,
and the erased-column criterion for correcting node erasures.
Section~\ref{sec:cyclic-slopes} specializes to the binary field and develops
the triple-node construction with an arbitrary cyclic slope and the multi-slope
extension of the slope-two construction to more failures.
Section~\ref{sec:optimal-triples}, also over the binary field,
constructs the ordinary-edge check matrix, called the edge skeleton, from
consecutive Frobenius powers, reduces Singleton optimality of the codes built on it to a common quotient of its check space
followed by a choice of loop columns, and gives the
Singleton-optimal codes at \(n=6,8,10,12\).
Section~\ref{sec:enumerators} returns to arbitrary prime powers and derives the
node-weight enumerators and the exact and asymptotic formulas for ball volumes, and
Section~\ref{sec:bounds} turns these volumes into packing, existence, and
covering bounds.
Section~\ref{sec:clique} treats the complementary clique metric, its exact
enumerator, and the node--clique duality, while
Section~\ref{sec:discussion} discusses the remaining construction and
enumeration problems.

\section{Definitions and Preliminaries}
\label{sec:preliminaries}

For a positive integer \(n\), let
\[
 [n]\triangleq\{0,1,\ldots,n-1\},
 \qquad
 V_n\triangleq\{v_0,v_1,\ldots,v_{n-1}\}.
\]
Let
\[
 E_n\triangleq\{\langle v_i,v_j\rangle:0\leq i\leq j<n\}
\]
be the set of all self-loops and ordinary unordered edges of the complete
graph on \(V_n\), and put
\[
 N\triangleq|E_n|=\binom{n+1}{2}.
\]
Throughout, \(q\) is a prime power.  The space of \(q\)-ary graph words is
\[
 \cG_q(n)\triangleq\F_q^{E_n}.
\]
For a graph word \(G\in\cG_q(n)\), let
\(e_{i,j}=e_{j,i}\in\F_q\) denote the label on
\(\langle v_i,v_j\rangle\).  Thus \(G\) can equivalently be displayed as the
symmetric labeling matrix
\[
 A_G\triangleq(e_{i,j})_{i,j\in[n]}.
\]
This is the labeling-matrix notation of
\cite{YohananovEfronYaakobi2020}.  Define
\[
\supp(G)
\triangleq
\{\langle v_i,v_j\rangle\in E_n:e_{i,j}\ne0\}.
\]
This set is the coordinate support of \(G\).  The associated support graph is
the looped undirected graph
\[
 \bigl(V_n,\supp(G)\bigr).
\]
Thus its vertex set is always all of \(V_n\), including isolated vertices,
and its edges are exactly the coordinates carrying nonzero labels.  Whenever
we speak of a vertex cover of a graph word, we mean a vertex cover of this
support graph.

\begin{definition}[Node weight and distance]
A vertex cover of \(G\) is an index set \(S\subseteq[n]\) such that
\(\{v_i:i\in S\}\) meets every edge of its support graph; if the loop at
\(v_i\) is nonzero, then \(i\in S\).  The node weight is
\[
 \wtN(G)\triangleq
 \min\{|S|:S\text{ is a vertex cover of }G\}.
\]
The node distance is
\[
 d_{\rm N}(G,G')\triangleq\wtN(G-G').
\]
\end{definition}

The node weight and node distance are the graph weight and graph distance of
\cite[Sec.~II]{YohananovYaakobi2019}; the word \emph{node} distinguishes them from
the clique weight of Section~\ref{sec:clique}.  The function \(d_{\rm N}\) is a
translation-invariant metric \cite{YohananovYaakobi2019}.  For a code
\(\cC\subseteq\cG_q(n)\), write \(d_{\rm N}(\cC)\) for its minimum distance.
Throughout, the minimum distance of a code containing at most one word is
taken to be \(+\infty\).

For \(S\subseteq[n]\), define its star-coordinate set and induced
clique-coordinate set by
\[
 \St(S)
 \triangleq
 \{\langle v_i,v_j\rangle\in E_n:\{i,j\}\cap S\neq\varnothing\}
\]
and
\[
 \Cl(S)
 \triangleq
 \{\langle v_i,v_j\rangle\in E_n:i,j\in S\},
\]
respectively.  Loops are included in both definitions.  The star coordinates
describe node failures; the complementary clique coordinates will be used in
Section~\ref{sec:clique}.  To match the notation of
\cite{YohananovYaakobi2019,YohananovEfronYaakobi2020}, where \(N_i\) is called
the neighborhood edge set of \(v_i\), for \(i\in[n]\) write
\[
 N_i\triangleq\St(\{i\}).
\]
Thus
\[
 \St(S)=\bigcup_{i\in S}N_i,
\]
and the erasure pattern generated by each failed vertex includes its loop
coordinate.  The node metric is also the combinatorial metric induced by the
coordinate cover
\(\{N_i:i\in[n]\}\) in the terminology of
\cite{PinheiroMachadoFirer2019}.  Define the star subspace
\[
 \cV_S
 \triangleq
 \{G\in\cG_q(n):\supp(G)\subseteq\St(S)\}.
\]
Then
\[
 \dim_{\F_q}\cV_S
 =
 |\St(S)|
 =
 |S|n-\binom{|S|}{2}.
\]

The distance characterization is \cite[Thm.~6]{YohananovYaakobi2019}.
We restate it in the present coordinate notation, together with its equivalent
star-subspace form.

\begin{lemma}[Distance and node erasures]
Let \(0\leq\rho\leq n\) be an integer.
A linear code \(\cC\subseteq\cG_q(n)\) corrects every set of at most
\(\rho\) node erasures if and only if
\[
 d_{\rm N}(\cC)\geq \rho+1.
\]
Equivalently,
\[
 \cC\cap\cV_S=\{0\}
 \qquad
 \text{for every }S\in\binom{[n]}{\rho}.
\]
\end{lemma}

The preceding condition has a direct parity-check formulation.  We state it
explicitly because the columns are indexed by graph coordinates, including
the loops.

\begin{definition}[Parity-check matrix]
Let \(\cC\subseteq\cG_q(n)\) be a linear code.  A parity-check matrix
for \(\cC\) is a matrix
\[
 H=(h_e)_{e\in E_n}\in\F_q^{s\times N}
\]
such that \(\cC=\ker H\).  Thus each ordinary edge and each loop indexes one
column of \(H\).  If the rows of \(H\) are linearly independent, then
\(s=N-\dim\cC\).  For \(X\subseteq E_n\), write \(H|_X\) for the submatrix
formed by the columns indexed by \(X\).
\end{definition}

\begin{lemma}[Erased-column criterion]
\label{lem:erased-column-criterion}
Let \(H\) be a parity-check matrix for a linear code \(\cC\subseteq\cG_q(n)\), and let
\(X\subseteq E_n\) be any set of erased coordinates.  The erasure of \(X\) is
uniquely correctable if and only if the columns of \(H|_X\) are linearly
independent.  In particular, the node-erasure pattern generated by
\(S\subseteq[n]\) is uniquely correctable if and only if the columns of
\(H|_{\St(S)}\) are linearly independent.  Consequently, \(\cC\) corrects
every set of at most \(\rho\) node erasures if and only if this holds for every
\(S\in\binom{[n]}{\rho}\).
\end{lemma}

\begin{proof}
Let \(x=(x_e)_{e\in X}\) be a labeling of the erased coordinates, and
extend it by zero on all surviving coordinates.  The resulting graph word
belongs to \(\cC=\ker H\) exactly when \(H|_Xx=0\).  Hence a nonzero
erased difference is invisible to all checks exactly when \(H|_X\) has a
nonzero kernel, equivalently, when its columns are linearly
dependent.
\end{proof}

For an integer \(0\leq\rho\leq n\), if a linear code
\(\cC\subseteq\cG_q(n)\) corrects every set of at most \(\rho\) node erasures,
then puncturing to the surviving looped clique gives the graph Singleton bound
\[
 \dim\cC\leq\binom{n-\rho+1}{2}.
\]
Indeed, if two codewords have the same restriction to that surviving clique,
then their difference is supported on the stars of the failed vertices.
Unique correction forces that difference to be zero, so the puncturing map is
injective and its target has the displayed number of coordinates.
Define the redundancy of a code \(\cC\) by
\[
 r(\cC)\triangleq N-\dim\cC.
\]
Define the graph Singleton redundancy by
\[
 r_{\rm Sing}(n,\rho)
 \triangleq
 \rho n-\binom{\rho}{2}.
\]
In these terms, the graph Singleton bound \cite[Eq.~(4)]{YohananovYaakobi2019} reads
\begin{equation}
 r(\cC)\geq r_{\rm Sing}(n,\rho).
\label{eq:singleton}
\end{equation}
A code correcting every set of at most \(\rho\) node erasures and meeting
\eqref{eq:singleton} with equality is called Singleton-optimal, or simply optimal.

The following proposition shows that a Singleton-optimal code for
\(\rho\) node failures on \(n\) vertices yields one for \(\rho-1\) failures on
\(n-1\) vertices; applied repeatedly, it turns an optimal family for any fixed
\(\rho\geq4\) into an optimal triple-node family (Section~\ref{sec:discussion}).  It is the
graph-coordinate specialization of puncturing, which contracts the
corresponding columns in the matroid represented by a parity-check matrix.

\begin{proposition}[Vertex descent for optimal codes]
\label{prop:vertex-descent}
If there is a Singleton-optimal linear \(\rho\)-node-erasure code on \(n\)
vertices, where \(2\leq\rho\leq n\), then there is a Singleton-optimal
\((\rho-1)\)-node-erasure code on \(n-1\) vertices over the same field.
\end{proposition}

\begin{proof}
Let \(H\) be a full-row-rank parity-check matrix for the original code, fix a
vertex \(v\), and let \(L_v\) be the span of the \(n\) columns on its loop
and incident ordinary edges.  Those columns are independent by
Lemma~\ref{lem:erased-column-criterion}, since they are contained in the
erased-column set of any \(\rho\)-set containing \(v\).
Let
\[
 \pi:\operatorname{im}H\longrightarrow\operatorname{im}H/L_v
\]
be the quotient map.  Delete the columns incident with \(v\), including its
loop.  Choose any basis of the quotient space
\(\operatorname{im}H/L_v\), express each remaining vector
\(\pi(h_e)\) in that basis, and form the matrix
\[
 \widetilde H\triangleq
 \bigl(\pi(h_e)\bigr)_{e\text{ a coordinate on the remaining vertices}}.
\]
Here the notation in the display means the coordinate column of
\(\pi(h_e)\) in the chosen quotient basis; changing that basis only applies
an invertible row operation.
For every \((\rho-1)\)-set \(S\) of remaining vertices, independence of the
erased columns of \(\widetilde H\) is equivalent to independence of the
original columns erased by \(S\cup\{v\}\).  Hence the descended code
\(\widetilde\cC\triangleq\ker\widetilde H\) corrects \(\rho-1\) node
erasures.  The columns of \(H\) span \(\operatorname{im}H\), and those
incident with \(v\) lie in \(L_v\), so the images of the remaining columns
span the quotient, and the redundancy of \(\widetilde\cC\) is the
quotient dimension
\[
 r_{\rm Sing}(n,\rho)-n
 =r_{\rm Sing}(n-1,\rho-1),
\]
so the descended code is optimal.
\end{proof}

\section{Cyclic-Slope Constructions}
\label{sec:cyclic-slopes}

This section and Section~\ref{sec:optimal-triples} specialize the \(q\)-ary
model of Section~\ref{sec:preliminaries} to the binary field.  We return to
arbitrary prime powers in Section~\ref{sec:enumerators}.

By Lemma~\ref{lem:erased-column-criterion}, recovery requires that the
columns on the erased coordinates be independent. For three failures we
keep the published neighborhood and diagonal checks and vary the cyclic
slope. We compute the exact nullity on each triple and characterize
recovery for primitive slopes by a projective arc. For more failures we
use several slopes, each giving an additional equation on the erased labels.

We first retain the notation of
\cite{YohananovEfronYaakobi2020}.  Let \(n\) be an odd prime.  For
congruence notation, write
\(\mathbb Z_n\triangleq\mathbb Z/n\mathbb Z\), with unit group
\(\mathbb Z_n^\times\), and identify \([n]\) with its canonical residue
representatives in \(\mathbb Z_n\).  For \(u\in\mathbb Z_n^\times\), write
\(\operatorname{ord}_n(u)\) for its multiplicative order modulo \(n\).  For
\(h,m\in[n]\), define
\begin{align}
 S_h
 &\triangleq
 \{\langle v_h,v_\ell\rangle\in E_n:\ell\in[n]\setminus\{h\}\},
\label{eq:Sh}\\
 D_m
 &\triangleq
 \{\langle v_u,v_v\rangle\in E_n:u+v\equiv m\pmod n\}.
\label{eq:Dm}
\end{align}
For \(\lambda\in\mathbb Z_n^\times\setminus\{1\}\) and \(s\in[n]\), put
\begin{equation}
 T_s(\lambda)
 \triangleq
 \left\{
 \langle v_u,v_v\rangle:
 \begin{array}{c}
 0\le u<v<n,\\[-2pt]
 u+\lambda v\equiv s\pmod n\ \text{or}\\[-2pt]
 v+\lambda u\equiv s\pmod n
 \end{array}
 \right\}.
\label{eq:Tlambda}
\end{equation}
The order \(u<v\) only selects a canonical description of each undirected
edge; both displayed orientations are part of the check family.
Thus \(S_h\), which is \(N_h\) without its loop, and every \(T_s(\lambda)\)
contain only ordinary edges, whereas
each \(D_m\) contains the unique loop indexed by \(2^{-1}m\pmod n\).
The set \(T_s(2)\) is exactly the slope-two set \(T_s\) in the
published triple-node construction, where \(D_m\) and \(T_s\) represent
parity constraints on the diagonals with slopes one and two of the labeling
matrix.  Accordingly, we call the parity checks on \(S_h\), \(D_m\), and
\(T_s(\lambda)\) the neighborhood, diagonal, and slope checks, and
\(\lambda\) their cyclic slope.

\subsection{A primitive-slope triple-node construction}

For a prime \(n\geq5\) and
\(\lambda\in\mathbb Z_n^\times\setminus\{1\}\), let
\(\cC_3(\lambda)\subseteq\cG_2(n)\) be the code of graph words whose labels
satisfy
\begin{align}
 \sum_{\langle v_i,v_j\rangle\in S_h}e_{i,j}&=0
 &&(h\in[n]),\label{eq:lambda-S}\\
 \sum_{\langle v_i,v_j\rangle\in D_m}e_{i,j}&=0
 &&(m\in[n]),\label{eq:lambda-D}\\
 \sum_{\langle v_i,v_j\rangle\in T_s(\lambda)}e_{i,j}&=0
 &&(s\in[n]).
\label{eq:lambda-T}
\end{align}
Let
\[
 H_{n,\lambda}\in\F_2^{3n\times N}
\]
be the parity-check matrix with columns indexed by \(E_n\) and with rows
equal to the incidence vectors of the \(3n\) checks in
\eqref{eq:lambda-S}--\eqref{eq:lambda-T}.  Thus
\(\cC_3(\lambda)=\ker H_{n,\lambda}\).
For the slope \(2\), set
\begin{equation}
 \cC_3\triangleq\cC_3(2).
\label{eq:literal-specialization}
\end{equation}
This is exactly the construction of \cite{YohananovEfronYaakobi2020}.

We now ask when these checks recover a fixed erased triple.  For
\(c\in\mathbb Z_n^\times\) and \(d\in\mathbb Z_n\), the affine permutation
\[
 \varphi_{c,d}:\mathbb Z_n\longrightarrow\mathbb Z_n,
 \qquad
 \varphi_{c,d}(j)\triangleq cj+d,
\]
induces a permutation of graph coordinates and sends the three check families
to
\[
 S_h\mapsto S_{ch+d},\qquad
 D_m\mapsto D_{cm+2d},\qquad
 T_s(\lambda)\mapsto T_{cs+(1+\lambda)d}(\lambda).
\]
It therefore preserves recoverability.  We may normalize the failed-index
set to
\[
 I\triangleq\{0,a,b\},
 \qquad
 a,b\in\mathbb Z_n^\times,\quad a\ne b.
\]
By Lemma~\ref{lem:erased-column-criterion}, this triple is recoverable exactly
when \(H_{n,\lambda}|_{\St(I)}\) has full column rank.

The slope checks link pairs of cyclic frequencies. To index these
frequencies, let \(\widetilde K\) be a splitting field of \(x^n-1\)
over \(\F_2\), and fix a primitive \(n\)-th root
\(\zeta\in\widetilde K\). If \(c_m\) is the sum of the labels in
check \(m\) of a cyclic family, the syndrome of that family at frequency \(k\) is
\(\sum_{m\in\mathbb Z_n}c_m\zeta^{km}\). As \(k\) ranges over
\(\mathbb Z_n\), these are the values of \(\sum_{m}c_mx^m\) at the \(n\) distinct
roots of \(x^n-1\);
they all vanish exactly when every \(c_m\) is zero.
For \(k\in\mathbb Z_n^\times\) and \(j\in\mathbb Z_n\), define
\[
 p_k(j)\triangleq
 \bigl(1,\zeta^{kj},\zeta^{k\lambda j}\bigr)^{\mathsf T}
 \in\widetilde K^3.
\]
For the normalized failed set \(I=\{0,a,b\}\), assemble the three vectors
into
\[
 M_k(a,b)\triangleq
 \bigl(p_k(0)\ \ p_k(a)\ \ p_k(b)\bigr)\in\widetilde K^{3\times3},
\]
and define the local determinant
\[
 \Delta_k(a,b)
 \triangleq
 \det M_k(a,b).
\]
Thus each failed vertex \(j\) gives one column: its entries are the
character values \(\zeta^{hj}\) at \(h=0,k,\lambda k\).
These are not columns of \(H_{n,\lambda}\): their determinant tests the
relation between the frequencies \(k\) and \(\lambda k\) after the neighborhood and diagonal
checks have been imposed. We show after Theorem~\ref{thm:primitive-slope}
that, once these checks hold, a word supported on the erased stars has one
unknown \(y_k\) at each nonzero frequency \(k\), and the slope checks read
\(\Delta_k(a,b)(y_k+y_{\lambda k})=0\); the frequencies are therefore linked
along the orbits of multiplication by \(\lambda\).
Since \(M_{2k}(a,b)\) is the entrywise square of \(M_k(a,b)\), one has
\(\Delta_{2k}(a,b)=\Delta_k(a,b)^2\); hence the indices \(k\) with
\(\Delta_k(a,b)=0\) form a union of cosets of the subgroup \(\langle2\rangle\)
of \(\mathbb Z_n^\times\), so their number is a multiple of
\(\operatorname{ord}_n(2)\); since \(\operatorname{ord}_n(2)\geq3\) for \(n\geq5\), this number is never one.
Appendix~\ref{app:primitive-slope} proves the following criterion.

\begin{lemma}[Local determinant criterion]
\label{lem:local-determinant-criterion}
Let \(n\ge5\) be prime, \(\lambda\in\mathbb Z_n^\times\setminus\{1\}\),
and \(I=\{0,a,b\}\) with \(a,b\) distinct and nonzero. For each orbit
\(O\) of multiplication by \(\lambda\) on \(\mathbb Z_n^\times\), set
\[
 z(O)\triangleq\bigl|\{k\in O:\Delta_k(a,b)=0\}\bigr|.
\]
Then the nullity of \(H_{n,\lambda}|_{\St(I)}\) is
\[
 \sum_O\max\{1,z(O)\}-1.
\]
In particular, no triple is uniquely recoverable for a nonprimitive
slope. For a primitive slope the triple is recoverable if and only if
all its determinants are nonzero. If any determinant vanishes for a
primitive slope, the nullity is the number of vanishing determinants
minus one and is at least \(\operatorname{ord}_n(2)-1\).
\end{lemma}

The determinant condition has a compact geometric formulation.  The
projective plane \(\operatorname{PG}(2,\widetilde K)\) is the set of
one-dimensional subspaces of \(\widetilde K^3\); a point is represented by a
nonzero triple, up to
multiplication by a nonzero scalar.  For every
\(\lambda\in\mathbb Z_n\), repeatedly applying
\((X,Y,Z)\mapsto(X,\zeta Y,\zeta^\lambda Z)\) to \((1,1,1)\) gives the cyclic
orbit
\begin{equation}
 \Gamma_{n,\lambda}
 \triangleq
 \{(1,\zeta^j,\zeta^{\lambda j}):j\in\mathbb Z_n\}
 \subseteq\operatorname{PG}(2,\widetilde K).
\label{eq:cyclic-arc}
\end{equation}
We call this set an \(n\)-arc if no three of its points are collinear.  In
matrix language, every \(3\)-by-\(3\) minor of the displayed orbit columns is
then nonzero.  The determinant \(\Delta_k(a,b)\) is precisely the minor on the
three columns indexed by \(0,ka,kb\).  Lemma~\ref{lem:local-determinant-criterion}
therefore gives the following coding characterization.

\begin{theorem}[Primitive-slope characterization]
\label{thm:primitive-slope}
Let \(n\geq5\) be prime and let \(\lambda\) be primitive modulo \(n\).
Then \(\cC_3(\lambda)\) corrects every three-node erasure if and only if
\(\Gamma_{n,\lambda}\) is an \(n\)-arc.
\end{theorem}

\begin{proof}
Translating three orbit indices by \(j\) multiplies their minor by
\(\zeta^{(1+\lambda)j}\ne0\), and \(\Delta_k(a,b)\) is the minor on \(0,ka,kb\).
Hence the point set \(\Gamma_{n,\lambda}\) is an \(n\)-arc exactly when
\(\Delta_k(a,b)\ne0\) for all distinct \(a,b\in\mathbb Z_n^\times\) and all
\(k\in\mathbb Z_n^\times\).  Apply
Lemma~\ref{lem:local-determinant-criterion} to every erased triple after an
affine relabeling.
\end{proof}

The orbit \(\Gamma_{n,\lambda}\) need not be an \(n\)-arc for a primitive
slope. At \(n=7\) the primitive slopes are \(3\) and \(5\), and for each of
them exactly three of the six determinants \(\Delta_k(1,3)\),
\(k\in\mathbb Z_7^\times\), vanish. By
Lemma~\ref{lem:local-determinant-criterion}, two binary unknowns therefore
remain on the erased triple \(\{v_0,v_1,v_3\}\). Since no triple is
recoverable for a nonprimitive slope, no \(\cC_3(\lambda)\) corrects every
three-node erasure at \(n=7\).

To relate the three columns \(p_k(0),p_k(a),p_k(b)\) to the erased
labels, let these labels take values in the splitting
field \(\widetilde K\); extension of scalars preserves the rank of a
binary matrix. Attach a variable \(Z_v\) to vertex \(v\) and record the
label \(e_{u,v}\) as the coefficient of \(Z_uZ_v\) in the formal polynomial
\[
 P\triangleq\sum_{u\leq v}e_{u,v}Z_uZ_v.
\]
In particular, \(Z_v^2\) carries the loop at \(v\).
For a frequency \(k\), put
\(\mathbf t_k\triangleq(\zeta^{kv})_{v\in\mathbb Z_n}\).
An edge with endpoints \(u,v\) contributes \(e_{u,v}\zeta^{k(u+v)}\)
to \(P(\mathbf t_k)\). Grouping by \(u+v\) gives
\[
 P(\mathbf t_k)=
 \sum_{m\in\mathbb Z_n}\zeta^{km}
 \sum_{\langle v_u,v_v\rangle\in D_m} e_{u,v}.
\]
The loop at \(v\) contributes \(e_{v,v}\zeta^{2kv}\) exactly once.
For the slope checks use the polar form
\(B_P(z,z')\triangleq P(z+z')-P(z)-P(z')\).
In \(B_P(\mathbf t_k,\mathbf t_{\lambda k})\) an ordinary edge contributes
\(e_{u,v}(\zeta^{k(u+\lambda v)}+\zeta^{k(v+\lambda u)})\), and a square
contributes zero. The two exponents before multiplication by \(k\) are
distinct for \(u\ne v\) since \(\lambda\ne1\); they index exactly the
two checks \(T_s(\lambda)\) containing the edge.

For the failed set \(I=\{0,a,b\}\), change variables to
\(Y_v\triangleq Z_v+Z_0\) for \(v\ne0\), with \(Y_0\triangleq0\).
The coefficient of \(Z_0Y_v\) is the neighborhood check at \(v\):
each ordinary edge incident with \(v\) contributes its label once.
The coefficient of \(Z_0^2\) is the sum of all labels, which is the
sum of the diagonal checks. These coefficients vanish for a codeword.
The remaining polynomial in the \(Y_v\) vanishes under \(Z_0=Z_a=Z_b=0\),
since the word is supported on the stars of \(0,a,b\); as this substitution
sets \(Y_a=Y_b=0\) and \(Y_v=Z_v\) otherwise, every remaining monomial
contains \(Y_a\) or \(Y_b\). Hence
\[
 P=Y_aL_a+Y_bL_b
\]
for linear forms \(L_a,L_b\) in the \(Y_v\).
Write \(d_i(k)\triangleq Y_i(\mathbf t_k)=1+\zeta^{ki}\) for \(i\in\{a,b\}\). At each \(k\ne0\)
both \(d_a(k)\) and \(d_b(k)\) are nonzero. The diagonal equation \(P(\mathbf t_k)=0\) reads
\[
 d_a(k)L_a(\mathbf t_k)+d_b(k)L_b(\mathbf t_k)=0
\]
and therefore leaves one unknown \(y_k\):
\(L_a(\mathbf t_k)=d_b(k)y_k\) and \(L_b(\mathbf t_k)=d_a(k)y_k\).
Subtracting the column \(p_k(0)\) from \(p_k(a)\) and \(p_k(b)\)
expresses \(\Delta_k(a,b)\) as the minor with columns
\((d_i(k),d_i(\lambda k))^{\mathsf T}\), \(i=a,b\).
Substitution in the polar form therefore gives the slope check
\begin{equation}
 B_P(\mathbf t_k,\mathbf t_{\lambda k})
 =\Delta_k(a,b)(y_k+y_{\lambda k})=0 .
\label{eq:lambda-factorization}
\end{equation}
Thus a nonzero determinant \(\Delta_k(a,b)\) forces \(y_k=y_{\lambda k}\). Appendix~\ref{app:primitive-slope} shows that, before the slope checks are imposed, every tuple
\((y_k)_{k\ne0}\in\widetilde K^{\,n-1}\) occurs, and that two tuples represent the same word supported on the
erased stars exactly when they differ by a single common constant.
The neighborhood check at \(0\), which is the sum of the others, and the slope syndrome at
frequency zero, which counts every ordinary edge twice, vanish automatically.
Hence each orbit of multiplication by \(\lambda\), cut at the links \(k\to\lambda k\) with
\(\Delta_k(a,b)=0\), keeps one free value per component, and the common constant
removes one; this is the count in Lemma~\ref{lem:local-determinant-criterion}.

\begin{proposition}[Rank of the cyclic-slope check matrix]
\label{prop:primitive-slope-rank}
Let \(n\geq5\) be prime and let \(\lambda\in\mathbb Z_n^\times\setminus\{\pm1\}\).  Then
\begin{equation}
 \rank H_{n,\lambda}=3n-2.
\label{eq:lambda-rank}
\end{equation}
\end{proposition}

Thus the code redundancy is
\[
 r(\cC_3(\lambda))=\rank H_{n,\lambda}=3n-2.
\]
The matrix has \(3n\) rows. In a candidate row dependence the diagonal
coefficients vanish, since each loop lies in exactly one row, a diagonal one.
Summing the column relations, read for ordered pairs \((u,v)\) and trivially true
at \(u=v\), over each line \(L(u,v)=t\) in
\(\mathbb Z_n^2\), for the forms \(L=u,v,u+\lambda v,v+\lambda u\), which are
pairwise nonproportional because \(\lambda\ne\pm1\), shows that the remaining
coefficients are constant on each neighborhood or slope family; the two family
sums form a basis of the row dependencies
(Appendix~\ref{app:primitive-slope-rank}).

We next give three slopes for which \(\Gamma_{n,\lambda}\) is an \(n\)-arc at every
prime \(n>5\).  We first
record the elementary exponent involution that will transfer the arc
property to two of them from the slopes \(4\) and \(16\), which are squares
modulo \(n\) and hence never primitive.

\begin{lemma}[Exponent involution]
\label{lem:exponent-involution}
For every \(\lambda\in\mathbb Z_n\), the point sets
\(\Gamma_{n,\lambda}\) and \(\Gamma_{n,1-\lambda}\) are projectively
equivalent.
\end{lemma}

\begin{proof}
Interchanging the first two coordinates of \((1,z,z^\lambda)\), scaling by
\(z^{-1}\), and writing \(w\triangleq z^{-1}\) gives
\[
 (1,z,z^\lambda)
 \longmapsto
 (z,1,z^\lambda)
 \sim(1,z^{-1},z^{\lambda-1})
 =(1,w,w^{1-\lambda}).
\]
As \(z\) runs through the cyclic subgroup, so does \(w\).  Hence
\(\Gamma_{n,\lambda}\) and \(\Gamma_{n,1-\lambda}\) are projectively
equivalent.
\end{proof}

\begin{lemma}[Prime-subgroup Frobenius arcs]
\label{lem:prime-subgroup-arcs}
Let \(n\) be an odd prime, put \(d\triangleq\operatorname{ord}_n(2)\), and take a
positive integer \(e\) with \(d\nmid e\). If \(d\) is even or
\(\gcd(d,e)=1\), then \(\Gamma_{n,2^e}\) is an \(n\)-arc.
\end{lemma}

\begin{proof}
Write \(\mu_n\triangleq\langle\zeta\rangle\), so
\(\F_2(\mu_n)=\F_{2^d}\).  If three orbit points were collinear, after
normalizing one parameter to \(1\) there would be distinct
\(x,y\in\mu_n\setminus\{1\}\) such that
\[
 (1+x)(1+y^{2^e})+(1+y)(1+x^{2^e})=0.
\]
Consequently
\[
 \eta\triangleq\frac{1+x}{1+y}
 \qquad\text{satisfies}\qquad
 \eta^{2^e-1}=1,
\]
and hence
\[
 \eta\in\F_{2^e}\cap\F_{2^d}
 =\F_{2^{\gcd(d,e)}}.
\]
If \(\gcd(d,e)=1\), then \(\eta=1\), contradicting \(x\ne y\).
Otherwise \(d\) is even. Put \(Q\triangleq2^{d/2}\). Since
\(Q\equiv-1\pmod n\), raising \(1+x=\eta(1+y)\) to the \(Q\)-th power
and comparing with the original equality gives
\[
 \frac yx=\eta^{Q-1}.
\]
The left side is a nonidentity element of \(\mu_n\) and has order \(n\).
The right side lies in \(\F_{2^{\gcd(d,e)}}^\times\), which contains no
element of order \(n\): \(\gcd(d,e)<d=\operatorname{ord}_n(2)\).
\end{proof}

When \(\gcd(d,e)=1\), the lemma also follows from the classical translation ovals
(called translation ovoids in \cite{Payne1971}): for such \(e\), no three of
the points \((1,x,x^{2^e})\), \(x\in\F_{2^d}\), are collinear
\cite[pp.~328--329]{Payne1971}, and \(\Gamma_{n,2^e}\) consists of such
points. For even \(d\), the subgroup argument also covers exponents not
coprime to \(d\). For example, at \(n=769\) one has \(d=384\), since \(2^{128}\not\equiv1\) and
\(2^{192}\equiv-1\pmod{769}\). Taking
\(e=6\) gives an arc at slope \(64\) and, by
Lemma~\ref{lem:exponent-involution}, at slope \(1-64=-63\).
The latter has order \(768\) modulo \(769\), since
\((-63)^{256}\not\equiv1\) and \((-63)^{384}\equiv-1\), so \(\cC_3(-63)\) corrects every triple
by Theorem~\ref{thm:primitive-slope}, although none of the three slopes
\(2,-3,-15\) used below is primitive modulo \(769\).

\begin{corollary}[Three fixed slopes giving arcs]
\label{cor:arc-safe-slopes}
For every prime \(n>5\), each member of
\[
 \{2,-3,-15\}
\]
gives an \(n\)-arc.
\end{corollary}

\begin{proof}
Apply Lemma~\ref{lem:prime-subgroup-arcs} with \(e=1,2,4\). If \(d\) is
odd, each of these exponents is coprime to \(d\). The three
slopes \(2^e\) are \(2,4,16\), and the corresponding integers
\(2^e-1\) are \(1,3,15\), none divisible by a prime greater than five;
hence \(n\nmid2^e-1\), that is, \(d\nmid e\).
Lemma~\ref{lem:exponent-involution} transfers the conclusions for \(4\)
and \(16\) to \(1-4=-3\) and \(1-16=-15\).
\end{proof}

\begin{theorem}[Unconditional infinite triple-node family]
\label{thm:unconditional-triple}
The set
\[
 \mathcal P\triangleq
 \left\{
 n>5:\begin{array}{l}
 n\text{ is prime and at least one }\\[-2pt]
 \lambda\in\{2,-3,-15\}\text{ is primitive modulo }n
 \end{array}
 \right\}
\]
is infinite.  For each \(n\in\mathcal P\), choose such a primitive
\(\lambda_n\).  Then \(\cC_3(\lambda_n)\) corrects every three-node erasure and
\begin{equation}
 r(\cC_3(\lambda_n))=3n-2,
 \qquad
 \dim\cC_3(\lambda_n)=\frac{(n-1)(n-4)}2.
\label{eq:unconditional-parameters}
\end{equation}
\end{theorem}

\begin{proof}
Heath-Brown's three-integer theorem \cite[Theorem, p.~27]{HeathBrown1986} applies to
nonzero multiplicatively independent integers \(g_1,g_2,g_3\) when none of
\[
 g_1,g_2,g_3,-3g_1g_2,-3g_1g_3,-3g_2g_3,g_1g_2g_3
\]
is an integer square.  Take
\[
 (g_1,g_2,g_3)\triangleq(2,-3,-15).
\]
Multiplicative independence follows successively from the valuations at
\(2,5,3\).  The seven required integers are
\[
 2,-3,-15,18,90,-135,90,
\]
where \(90\) occurs twice because two of the seven expressions have the same
value in this specialization; none is a square.  Heath-Brown's theorem
therefore proves that, for
infinitely many primes, at least one of the three selected integers is a
primitive root.  Corollary~\ref{cor:arc-safe-slopes} and
Theorem~\ref{thm:primitive-slope} give correction, while
Proposition~\ref{prop:primitive-slope-rank} gives rank \(3n-2\).
Subtracting this rank from \(\binom{n+1}{2}\) gives the dimension in
\eqref{eq:unconditional-parameters}.
\end{proof}

The new family has the same redundancy and erasure capability as the published
construction \(\cC_3\).  Its contribution is arithmetic: allowing
the slope to be selected from three fixed integers gives codes at infinitely
many prime lengths unconditionally.  At each such length its
redundancy is one above the graph Singleton redundancy \(3n-3\), and two below
the redundancy \(3n\) of Schmidt's rank-metric construction.  For example,
\[
 \operatorname{ord}_{17}(2)=8,
 \qquad
 \operatorname{ord}_{17}(-3)=16,
\]
so \(\cC_3(-3)\) corrects every triple at \(n=17\). The exact nullity
also shows what goes wrong at slope two: for \(\lambda=2\) the columns
of \(M_k(a,b)\) are \((1,t,t^2)^{\mathsf T}\) at the distinct values
\(t=1,\zeta^{ka},\zeta^{kb}\), so every determinant is a nonzero
Vandermonde determinant. Hence \(z(O)=0\) for every orbit \(O\) of multiplication by \(2\), giving nullity
\((n-1)/\operatorname{ord}_n(2)-1\) on every triple. At \(n=17\) one
binary unknown therefore remains on each erased triple.

\begin{proposition}[Decoding three known node erasures]
\label{prop:deterministic-decoding}
Suppose \(n\ge5\) is prime, \(\lambda\) is primitive modulo \(n\), and
\(\Gamma_{n,\lambda}\) is an \(n\)-arc. Any three erased vertices of
\(\cC_3(\lambda)\) can be recovered deterministically in \(O(n^2)\)
operations over \(\F_2\), including the computation of the syndromes.
\end{proposition}

\begin{proof}
Translate the erased vertices to \(0,a,b\). We seek the formal
polynomial on their stars whose check values equal the contributions of
the known labels. Write these contributions as \(s_v,d_j,t_j\) for the
neighborhood, diagonal and slope checks, respectively. Put \(d_*\triangleq\sum_j d_j\).
With the variables \(Y_v=Z_v+Z_0\) used above, first subtract the checks of
\[
 P_0\triangleq d_*Z_0^2+\sum_{v\ne0}s_v Z_0Y_v .
\]
Every monomial of \(P_0\) in the original variables lies on the star of
vertex zero. It supplies the required neighborhoods (also at zero, since
\(\sum_v s_v=0\)) and the sum of the diagonals. Thus the polynomial still
to be recovered has the form \(Y_aL_a+Y_bL_b\).

We do the calculation over binary polynomials, without finding any
roots. Put \(R\triangleq\F_2[x]/(1+x+\cdots+x^{n-1})\),
\(\delta_i\triangleq1+x^i\), and \(\sigma(g)(x)\triangleq g(x^\lambda)\), taking exponents
modulo \(n\). Evaluation at \(Z_v=x^v\) identifies the span of the
\(Y_v\), \(v\ne0\), with \(R\): the length-\(n\) coefficient vectors of
even parity form a complete set of representatives. Denote the images
of \(L_a,L_b\) by \(A,B\), and pack the residual diagonal and slope checks, that is, the values
\(d'_j,t'_j\) left after subtracting the checks of \(P_0\),
as \(d\triangleq\sum_j d'_j x^j,\ t\triangleq\sum_j t'_j x^j\) in \(R\).
Evaluation at \((x^v)_v\), and polarization at \((x^v)_v\) and
\((x^{\lambda v})_v\), give
\[
 d=\delta_a A+\delta_b B,\qquad
 t=\delta_a\sigma(A)+\sigma(\delta_a)A+
   \delta_b\sigma(B)+\sigma(\delta_b)B .
\]
No check is lost at \(x=1\): the residual diagonals sum to zero, and
every ordinary edge contributes twice to the sum of the slope checks.

Set \(A_0\triangleq d/\delta_a\). All solutions of the first equation are
\(A=A_0+\delta_b y,\ B=\delta_a y\). Substitution in the second gives
\[
 y+\sigma(y)=
 \frac{t+\delta_a\sigma(A_0)+\sigma(\delta_a)A_0}
      {\delta_a\sigma(\delta_b)+\delta_b\sigma(\delta_a)} .
\]
Both divisions are legal: \(\delta_a\) has no nontrivial \(n\)-th root
of unity as a zero, and at \(x=\zeta^k\) the second denominator is
\(\Delta_k(a,b)\ne0\) for every \(k\ne0\). These tests ensure
invertibility in \(R\) even if its modulus is reducible.

Write \(h_i\) and \(w_i\) for the coefficients of the even representatives of the
right-hand side and of \(y\). Start with \(w_1\triangleq0\), then set
\(w_{\lambda i}\triangleq h_{\lambda i}+w_i\) along the single cycle of nonzero
indices. The equation at zero and the last equation of the cycle
hold because the syndromes came from a codeword. Choose \(w_0\) for
even parity. The other starting value gives the same formal polynomial:
it adds \(1\) in \(R\) to \(y\), hence adds \(Y_b,Y_a\) to \(L_a,L_b\).
Read \(L_a,L_b\) from the nonzero-index coefficients of the even
representatives of \(A,B\). Expanding \(P_0+Y_aL_a+Y_bL_b\) gives the
erased labels, including the coefficients of the squares for the loops.

Computing the check values reads \(O(n^2)\) labels. Classical polynomial
multiplication and the extended Euclidean algorithm take \(O(n^2)\)
binary operations for the bounded number of products and divisions in
\(R\). The index cycle and the final expansion each take \(O(n)\) operations.
\end{proof}

The quadratic running time was already noted for slope two in
\cite[Sec.~VI]{YohananovEfronYaakobi2020}. When the input consists of
individual surviving labels, an exact decoder for any slope in the proposition
must read all but at most one of them. Indeed, there are \(u\triangleq\binom{n-2}{2}\)
surviving labels. Projection onto them is injective and the code has
dimension \(u-1\), so its image is a hyperplane in \(\F_2^u\). If two
labels were left unread, their coordinate plane would meet this hyperplane
in a nonzero vector. The corresponding codeword cannot vanish on the
erased labels: it would then be supported on just two edges or fewer,
of node weight at most two. Adding that codeword preserves every queried
label but changes the required output. Thus even adaptive queries require
at least \(u-1\) reads, and a basis of \(u-1\) surviving coordinate
functionals attains this read count.

For slope two, Yohananov, Efron, and Yaakobi
\cite[Sec.~VI, (17)]{YohananovEfronYaakobi2020} give a sparse system for
the three erasures and refer to Wiedemann's solver \cite{Wiedemann1986}.
Here the check equations can be solved in \(O(n)\) operations once the
syndromes are known: \(\sigma(g)=g^2\), so the denominator in the
equation for \(y\) is
\(\delta_a\delta_b(\delta_a+\delta_b)=x^a\delta_a\delta_b\delta_{b-a}\).
Division by each binomial \(\delta_j\) takes one traversal of the additive
cycle \(i\mapsto i+j\) on the even representatives: use
\(f_i+f_{i-j}=g_i\), fixing the free starting value by even parity.
All other products use cyclic shifts and additions. Binomial division of
this kind is also used in array-code decoding \cite[Sec.~V-A]{HouHanShumLi2018}.

For the published slope-two code, Appendix~\ref{app:one-bit} characterizes
when its dimension can be increased by one while retaining triple-node
correction, and certifies by exact computation that no such extension exists at
\(n=5,11,13,19\).

\subsection{The multi-slope construction}

Throughout this subsection, let \(n\) be an odd prime for which \(2\) is
primitive modulo \(n\).  For each \(2\leq\rho<n\), this hypothesis ensures that
\(2^k\not\equiv1\pmod n\) for every \(1\leq k\leq\rho-2<n-1\).  Define
\begin{equation}
 T_t^{(k)}\triangleq T_t(2^k),
 \qquad
 t\in\mathbb Z_n,\quad 1\leq k\leq\rho-2.
\label{eq:Tks}
\end{equation}

\begin{definition}[Multi-slope code]
For \(2\leq\rho<n\), the binary code
\(\cC_\rho^{\rm MS}\subseteq\cG_2(n)\) consists of all graph words
whose labels satisfy
\begin{align}
 \sum_{\langle v_i,v_j\rangle\in S_h}e_{i,j}&=0
 &&(h\in[n]),\label{eq:MS-S}\\
 \sum_{\langle v_i,v_j\rangle\in D_m}e_{i,j}&=0
 &&(m\in[n]),\label{eq:MS-D}\\
 \sum_{\langle v_i,v_j\rangle\in T_t^{(k)}}e_{i,j}&=0
 &&(t\in[n],\ 1\leq k\leq\rho-2).
\label{eq:MS-T}
\end{align}
\end{definition}

Let \(H_\rho^{\rm MS}\in\F_2^{\rho n\times N}\) be the matrix whose columns
are indexed by \(E_n\) and whose rows, in the order displayed above, are the
incidence vectors of the checks in
\eqref{eq:MS-S}--\eqref{eq:MS-T}.  Thus
\[
 \cC_\rho^{\rm MS}=\ker H_\rho^{\rm MS}.
\]
Its redundancy is the rank of these \(\rho n\) rows.

At \(\rho=3\), this definition gives
\[
 \cC_3^{\rm MS}=\cC_3(2)=\cC_3.
\]

\begin{theorem}[Multi-slope correction]
\label{thm:multislope}
Suppose that \(n\) is an odd prime, \(2\) is primitive modulo \(n\), and
\[
 2\leq\rho\leq\frac{n+1}{2}.
\]
Then the code \(\cC_\rho^{\rm MS}\) corrects every set of at most \(\rho\)
node erasures, and
\begin{equation}
 r(\cC_\rho^{\rm MS})=\rho n-(\rho-1).
\label{eq:MS-rank}
\end{equation}
\end{theorem}

For \((n+3)/2\leq\rho<n\), the same construction gives only the zero code.
Indeed, \(T_t(\lambda^{-1})=T_{\lambda t}(\lambda)\) for every slope \(\lambda\)
and every \(t\), so the slope families of \(2^k\) and \(2^{n-1-k}\) coincide.
Since the code uses every \(2^k\) with \(1\leq k\leq(n-1)/2\leq\rho-2\), it equals
the code whose slope checks use all \(2^k\) with \(1\leq k\leq n-2\). The
correction argument below and in Appendix~\ref{app:multislope}, applied to the
failed set \(A\) of all \(n\) vertices, so that \(s=|A|-1=n-1\), shows that this code is
zero.

We now use the formal polynomial over \(\F_2\) itself. Fix a failed set \(A\triangleq\{a_0,\ldots,a_s\}\) of \(\rho\) vertices, with
\(s\triangleq\rho-1\), and let a binary codeword be supported on its stars. Write
its polynomial as \(P=\sum_{u\leq v}e_{u,v}Z_uZ_v\). Center the variables
at \(a_0\), putting \(Y_v\triangleq Z_v+Z_{a_0}\) for \(v\ne a_0\) and
\(Y_{a_0}\triangleq0\). The neighborhood equations kill the coefficients of
\(Z_{a_0}Y_v\), and the sum of the diagonal equations kills that of
\(Z_{a_0}^2\). By the support calculation at the start of
Appendix~\ref{app:primitive-slope}, every surviving monomial contains the
variable of a failed vertex other than \(a_0\), so for suitable binary linear
forms
\(L_i\) in the \(Y_v\),
\[
 P=\sum_{i=1}^s Y_{a_i}L_i .
\]
Put \(F\triangleq\F_2(\zeta)\), where \(\zeta\) has order \(n\), and write
\(\mathbf t\triangleq(\zeta^v)_{v\in\mathbb Z_n}\).
The field has binary dimension \(n-1\), since \(2\) is primitive modulo
\(n\). Define \(\xi_i\triangleq Y_{a_i}(\mathbf t)\) and \(\eta_i\triangleq L_i(\mathbf t)\).
The diagonal evaluation at \(\mathbf t\) is \(P(\mathbf t)=0\). For slope
\(2^k\), evaluate the polar form at \(\mathbf t\) and its entrywise
\(2^k\)-th power. Since the forms \(L_i\) have binary coefficients,
the resulting equations are
\[
 \sum_i\xi_i\eta_i=0,\qquad
 \sum_i(\xi_i\eta_i^{2^k}+\xi_i^{2^k}\eta_i)=0
 \quad(1\le k<s).
\]
The first equation uses loops as well as ordinary edges; the others use
only ordinary edges. Appendix~\ref{app:multislope} shows that evaluation at \(\mathbf t\)
is injective on the binary span of the \(Y_v\), so the \(\xi_i\) are binary
independent, and proves by a Moore-matrix argument
(Lemma~\ref{lem:frobenius-descent}) that together
these equations force \(L_i=\sum_jC_{ji}Y_{a_j}\) for a symmetric binary matrix \(C\) with zero
diagonal, so that the polynomial \(P\), and hence the supported codeword, is zero;
it also derives \eqref{eq:MS-rank}, the only assertion whose proof uses
\(\rho\leq(n+1)/2\), from the rank count of
Appendix~\ref{app:primitive-slope-rank}, which applies because under this bound the
slopes \(2^k\), \(1\leq k\leq\rho-2\), avoid \(\pm1\) and contain no inverse pair.

\begin{corollary}[Redundancy comparison]
\label{cor:comparison}
Suppose that \(n\) is an odd prime, \(2\) is primitive modulo \(n\), and
\[
 2\leq\rho\leq\frac{n+1}{2}.
\]
Then
\[
 r(\cC_\rho^{\rm MS})-r_{\rm Sing}(n,\rho)
 =
 \binom{\rho-1}{2}.
\]
Schmidt's rank-metric construction for \(\rho<(n+1)/2\), and the zero code
for \(\rho=(n+1)/2\), have redundancy
\[
 r_{\rm Schmidt}(n,\rho)\triangleq\rho n,
\]
so the multi-slope construction saves \(\rho-1\) binary checks.
\end{corollary}

\begin{proof}
Under the labeling-matrix identification \(G\mapsto A_G\) from
Section~\ref{sec:preliminaries}, a graph word supported on the union of
\(\rho\) coordinate stars is a symmetric matrix.  After ordering the failed
vertices first, it has block form
\[
 \begin{pmatrix}A_{\rm ff}&B_{\rm fh}\\
 B_{\rm fh}^{\mathsf T}&0\end{pmatrix},
 \qquad
 A_{\rm ff}\in\F_2^{\rho\times\rho},
 \quad
 B_{\rm fh}\in\F_2^{\rho\times(n-\rho)},
\]
where \(A_{\rm ff}\) contains the labels of loops and edges with both endpoints
failed and
\(B_{\rm fh}\) contains labels between a failed and a healthy vertex.  The
zero block contains the healthy--healthy coordinates, which lie outside the
failed stars.  The full matrix therefore has rank at most \(2\rho\).
Thus a binary symmetric rank-metric code of minimum rank distance \(2\rho+1\)
corrects the node erasures, although it solves the stronger problem of
excluding every nonzero codeword difference of rank at most \(2\rho\).
For odd \(n\), the even-characteristic construction of
\cite{Schmidt2010} (see also \cite[Thm.~5.3]{Schmidt2020}) has dimension
\[
 k_{\rm Schmidt}=\frac{n(n-2\rho+1)}2
\]
for \(2\leq\rho<(n+1)/2\), and hence redundancy \(\rho n\); as node-erasure
codes these are the codes of \cite[Constr.~5, Thm.~11]{YohananovYaakobi2019}.
At
\(\rho=(n+1)/2\), the required minimum rank distance \(2\rho+1=n+2\) exceeds \(n\), so only the zero
code qualifies; its redundancy is \(N=\rho n\).  Subtraction gives the two asserted gaps.
\end{proof}

For \(\rho=2\), the code is the published optimal double-node code of
\cite{YohananovEfronYaakobi2020}, with redundancy \(2n-1\).  For
\(\rho=3\), it is \(\cC_3\), with redundancy
\(3n-2\), one bit above the graph Singleton redundancy \(3n-3\).

\section{A Uniform Frobenius--Moore Edge Skeleton and Finite Triple Completions}
\label{sec:optimal-triples}

All constructions in this section are binary and remain in the original
looped space \(\cG_2(n)\).  Fix an integer \(2\leq\rho\leq n\) and a failed
vertex-index set \(Q\in\binom{[n]}{\rho}\).  A full-row-rank parity-check matrix
of a Singleton-optimal code has
\[
 r_{\rm Sing}(n,\rho)=\rho n-\binom\rho2
\]
rows, and, by Lemma~\ref{lem:erased-column-criterion}, all columns indexed
by \(\St(Q)\) must be independent.  Exactly
\(\rho\) of these erased columns are loops; the remaining
\[
 \rho(n-1)-\binom\rho2
\]
are ordinary edges meeting \(Q\).  We
first set the loop columns aside and build one ordinary-edge check matrix,
the edge skeleton, that has full local rank for every \(Q\).  We then return to the complete coding
problem for \(\rho\geq3\): we seek one quotient of the check space, of dimension \(r_{\rm Sing}(n,\rho)\),
that preserves these edge ranks for every \(Q\). In that quotient the
\(\rho\) loops at each failed set must fill the \(\rho\) remaining
directions. Theorem~\ref{thm:uniform-moore-skeleton}
solves the first stage, Proposition~\ref{prop:loop-completion} gives the exact
completion criterion, and Proposition~\ref{prop:finite-completions} supplies
optimal triple-node codes at four finite lengths.
After relabeling the vertices by an affine frame below, we reuse \(Q\) for
the corresponding \(\rho\)-subset of field labels.

We now construct the ordinary-edge part of the desired parity-check matrix.
The name Moore recalls the Moore matrix of finite-field theory, whose rows are
consecutive Frobenius powers; the blocks of the column formula below use the
consecutive Frobenius powers \(x\mapsto x^{2^k}\), \(1\leq k\leq\rho-1\).  The Frobenius alternating-form
layers used below are classical.  Over finite
fields their rank properties go back to Delsarte and Goethals
\cite{DelsarteGoethals1975}; in the notation closest to ours, see
\cite[Eq.~(29) and Thm.~14]{Schmidt2010}.  Gow and Quinlan give the related
Galois-theoretic decomposition of the space of alternating forms and, for
cyclic extensions of odd and of even degree, the corresponding rank bounds
\cite{GowQuinlan2009}.  We build on
those layers: we realize their evaluations as columns indexed by the ordinary
edges of an affine frame and prove that, for every failed set, the erased
ordinary-edge columns are linearly independent.

Put
\[
 m\triangleq n-1,
 \qquad
 F\triangleq\F_{2^m},
\]
and choose an affine frame
\[
 \cA\triangleq\{a_0,a_1,\ldots,a_m\}\subseteq F.
\]
By definition of an affine frame, from any fixed frame point, the \(m\) differences to the other frame
points form an \(\F_2\)-basis of \(F\).  We use these points as algebraic
labels for the vertices of the original graph, identifying
\[
 v_i\longleftrightarrow a_i,
 \qquad 0\leq i\leq m.
\]
Accordingly, the ordinary edge \(\langle v_i,v_j\rangle\) is denoted by
\(\{a_i,a_j\}\), and a subset of \(\cA\) denotes the corresponding vertex
set.  This is only a relabeling of the original graph coordinates; the field
structure will be used to assign a check column to each ordinary edge.

Fix an ordered \(\F_2\)-basis of \(F\).  Expanding each field element in
this basis and concatenating the resulting blocks identifies
\[
 F^\rho\cong\F_2^{\rho m}.
\]
Thus an element of \(F^\rho\) below represents one binary parity-check
column of length \(\rho m\), written compactly as \(\rho\) consecutive
\(m\)-bit blocks.  For every ordinary edge
\(\{a,b\}\in\binom{\cA}{2}\), assign the single column
\begin{equation}
 h_{a,b}^{(\rho)}
 \triangleq
 \begin{pmatrix}
  a+b\\
  a^2b+ab^2\\
  \vdots\\
  a^{2^{\rho-1}}b+ab^{2^{\rho-1}}
 \end{pmatrix}
 \in F^\rho\cong\F_2^{\rho m}.
\label{eq:Moore-column}
\end{equation}
The formula is symmetric in \(a\) and \(b\), so it is well defined for an
unordered edge.  Assemble these columns into the edge skeleton, the single
global ordinary-edge check matrix
\[
 H_{\rm edge}^{(\rho)}
 \triangleq
 \bigl(h_{a,b}^{(\rho)}\bigr)_{
   \{a,b\}\in\binom{\cA}{2}}
 \in\F_2^{\rho m\times\binom n2}.
\]
This is the ordinary-edge portion from which the full parity-check matrix
will be obtained.

A cyclic frame relates these blocks to the checks of
Section~\ref{sec:cyclic-slopes}. When \(n\) is prime and \(2\) is primitive
modulo \(n\), let \(\zeta\in F\) have order \(n\). Its minimal polynomial is
\(1+x+\cdots+x^{n-1}\), so the only binary linear relation among
\(\zeta^0,\ldots,\zeta^{n-1}\) is \(\sum_j\zeta^j=0\), with an odd number of
terms; hence the points \(a_j\triangleq\zeta^j\), \(j\in[n]\), form an affine
frame. For \(\rho<n\), the blocks \(\zeta^u+\zeta^v\) and
\(\zeta^{u+2^kv}+\zeta^{v+2^ku}\), \(1\le k\le\rho-1\), of
\(h_{a_u,a_v}^{(\rho)}\) are then the coefficients of \(e_{u,v}\) in the
frequency-one syndromes of the neighborhood checks and of the slope checks
\(T_s(2^k)\).

For \(Q\in\binom{\cA}{\rho}\), define the ordinary-edge part of the erased
star by
\[
 \Stedge(Q)
 \triangleq
 \left\{\{a,b\}\in\binom{\cA}{2}:\{a,b\}\cap Q\ne\varnothing\right\}.
\]
The corresponding erased-edge block is
\[
 \left.H_{\rm edge}^{(\rho)}\right|_{\Stedge(Q)},
\]
and it contains
\[
 |\Stedge(Q)|
 =\rho(n-1)-\binom\rho2
 =\rho m-\binom\rho2
\]
columns.  The purpose of the algebraic assignment
\eqref{eq:Moore-column} is to give this block full column rank
simultaneously for every failed set \(Q\).

For a fixed vertex label \(a\), write \(x\triangleq a+b\) for the direction from
\(a\) to the other endpoint \(b\).  Since \(b=a+x\), substituting in
\eqref{eq:Moore-column} gives the binary linear map
\begin{equation}
 L_a^{(\rho)}:F\longrightarrow F^\rho,
 \qquad
 x\longmapsto
 \begin{pmatrix}
 x\\
 a^2x+ax^2\\
 \vdots\\
 a^{2^{\rho-1}}x+ax^{2^{\rho-1}}
 \end{pmatrix}.
\label{eq:Ua}
\end{equation}
Define
\[
 U_a^{(\rho)}\triangleq\operatorname{im}L_a^{(\rho)}\leq F^\rho.
\]
Each ordinary-edge column of
\(H_{\rm edge}^{(\rho)}\) incident with \(a\) is the image of the corresponding
basis direction \(x=a+b\).  The first block of \(L_a^{(\rho)}(x)\) is \(x\), so
\(L_a^{(\rho)}\) is injective.  Since these \(n-1\) directions form a basis of
\(F\), the incident columns form a basis of \(U_a^{(\rho)}\).
For several failed vertices we must also show that their incident columns,
with each erased edge counted once, are linearly independent.

\begin{theorem}[Uniform Moore edge skeleton]
\label{thm:uniform-moore-skeleton}
For every \(2\leq\rho\leq n\) and every \(\rho\)-subset
\(Q\subseteq\cA\), the erased-edge block has full column rank:
\[
 \rank_{\F_2}
 \left(\left.H_{\rm edge}^{(\rho)}\right|_{\Stedge(Q)}\right)
 =|\Stedge(Q)|
 =\rho m-\binom\rho2.
\]
Equivalently,
\[
 \dim_{\F_2}\sum_{a\in Q}U_a^{(\rho)}
 =\rho m-\binom\rho2
 =\rho n-\binom{\rho+1}{2}.
\]
\end{theorem}

Fix \(Q\in\binom{\cA}{\rho}\), and suppose that a binary relation among the
columns of its erased-edge block is
\begin{equation}
 \sum_{\{a,b\}\in\Stedge(Q)}
 c_{a,b}h_{a,b}^{(\rho)}=0,
 \qquad c_{a,b}\in\F_2.
\label{eq:Moore-proposed-relation}
\end{equation}
Here \(c_{a,b}\) is the coefficient of the single column for edge
\(\{a,b\}\). Choose one failed vertex as the anchor and give the other
failed vertices the first positions in an ordering of the other labels.
Assign each edge not incident with the anchor to its earlier endpoint. In a
relation supported on the erased edges, this endpoint must be failed.
Grouping the Frobenius blocks by that endpoint gives the paired sums in
the following lemma. The lemma first rules out edges joining a nonanchor failed
vertex to a healthy vertex. The coefficients between the nonanchor
failed vertices must then form a symmetric, strictly triangular matrix,
which is zero.

These Frobenius sums are classical. Over finite fields the corresponding
alternating trace forms are studied in
\cite[Thms.~3 and~14]{Schmidt2010} and \cite{GowQuinlan2009}. Over any
field of characteristic two the following lemma and its product variant
in Lemma~\ref{lem:frobenius-descent} are coordinate consequences of
Quillen's theorem on common Frobenius zeros \cite[Thm.~2.4]{Quillen1971},
for his forms (2.11) and (2.13), respectively. We give the coordinate
argument using successive Frobenius rows, as in Quillen's proof.

\begin{lemma}[Shifted Frobenius descent]
\label{lem:shifted-descent}
Let \(F\) be a field of characteristic two, and let \(s\geq1\) be an integer.
Suppose that
\(\xi_1,\ldots,\xi_s\in F\) are linearly independent over \(\F_2\) and that
\(\eta_1,\ldots,\eta_s\in F\) satisfy
\[
 \phi_k\triangleq\sum_{i=1}^s
 \left(\xi_i\eta_i^{2^k}+\xi_i^{2^k}\eta_i\right)=0,
 \qquad 1\leq k\leq s.
\]
Then there is a symmetric binary matrix \(C\) of size \(s\) such that
\(\eta_i=\sum_{j=1}^s C_{ji}\xi_j\) for every \(i\).
\end{lemma}

\begin{proof}
Write \(x_j\triangleq(\xi_1^{2^j},\ldots,\xi_s^{2^j})\) and
\(y_j\triangleq(\eta_1^{2^j},\ldots,\eta_s^{2^j})\) as row vectors. Let \(X,Y\)
have rows \(x_j,y_j\), respectively, for \(0\leq j<s\).
The Moore matrix \(X\) is invertible: a nonzero linear dependence among
its rows would give a nonzero polynomial \(\sum_{j<s}c_jT^{2^j}\) of
degree at most \(2^{s-1}\) vanishing on the entire binary span of the
\(\xi_i\), which has \(2^s\) elements.

Let \(C\triangleq X^{-1}Y\), so \(y_j=x_jC\) for \(0\leq j<s\).
For \(0\leq i<j<s\), the \((i,j)\) entry of \(XY^{\mathsf T}+YX^{\mathsf T}\)
is \(\phi_{j-i}^{\,2^i}=0\); the diagonal is zero in characteristic two.
Thus \(X(C+C^{\mathsf T})X^{\mathsf T}=0\), and \(C\) is symmetric.
Write \(C^{[2]}\) for its entrywise square. Squaring the preceding row
identities gives
\begin{equation}
 x_j(C+C^{[2]})=0\qquad(1\leq j<s).
\label{eq:descent-common-rows}
\end{equation}
These \(s-1\) rows are independent, since the matrix with rows
\(x_1,\ldots,x_s\) is the entrywise square of \(X\). Up to this point
only \(\phi_1,\ldots,\phi_{s-1}\) have been used.

The last equation \(\phi_s=0\) gives the same row identity at \(j=s\):
for every \(0\leq i<s\),
\[
 0=\phi_{s-i}^{\,2^i}
   =x_i y_s^{\mathsf T}+y_i x_s^{\mathsf T}
   =x_i(y_s^{\mathsf T}+C x_s^{\mathsf T}).
\]
As \(X\) is invertible and \(C\) is symmetric, \(y_s=x_s C\).
Squaring \(y_{s-1}=x_{s-1}C\) gives \(y_s=x_s C^{[2]}\) as well.
Now all \(s\) independent rows \(x_1,\ldots,x_s\) annihilate
\(C+C^{[2]}\), so \(C=C^{[2]}\) has entries in \(\F_2\).
The first row identity \(y_0=x_0C\) gives the stated coefficients.
\end{proof}

\begin{proof}[Proof of Theorem~\ref{thm:uniform-moore-skeleton}]
Choose one failed vertex as the anchor. Translating its field label to
zero preserves the rank: translating all labels by \(t\) keeps the first
block \(x=a+b\) unchanged and adds \(t^{2^k}x+tx^{2^k}\) to block \(k\),
an invertible binary shear of the check space. The other labels are now
a binary basis \(e_1,\ldots,e_m\); order them so that \(e_1,\ldots,e_s\)
are failed, where \(s\triangleq\rho-1\).

In a relation \eqref{eq:Moore-proposed-relation}, write \(c_{ij}\) for the
coefficient of the edge \(\{e_i,e_j\}\), setting it to zero on surviving
edges. Assign every nonanchor edge in the relation to its endpoint of
smaller index. This endpoint is failed, since all failed labels come first.
For example, with failures \(0,e_1,e_2\), the edge \(\{e_1,e_2\}\)
is counted at \(e_1\), not at both endpoints.
For \(1\le i\le s\), collect the coefficients of the edges assigned to \(e_i\)
in the field element
\[
 \eta_i\triangleq\sum_{j=i+1}^{m}c_{ij}e_j .
\]
The edges at the anchor contribute zero to every Frobenius block.
The sum in block \(k\) of the relation is therefore
\[
 \sum_{i=1}^s(e_i^{2^k}\eta_i+e_i\eta_i^{2^k})=0,
 \qquad 1\le k\le s.
\]
Lemma~\ref{lem:shifted-descent} gives a symmetric binary matrix \(C\)
with \(\eta_i=\sum_{j=1}^s C_{ji}e_j\).
Comparing coefficients in the full basis first gives \(c_{ij}=0\)
for \(j>s\). For \(j\le s\), the definition of \(\eta_i\) gives
\(C_{ji}=0\) whenever \(j\le i\). By symmetry \(C=0\), so all
nonanchor edge coefficients vanish. In the first block only the edges
\(\{0,e_j\}\) remain. Their first-block columns are the basis vectors
\(e_j\), so their coefficients vanish as well. Thus all the erased
ordinary-edge columns are independent.
\end{proof}

Theorem~\ref{thm:uniform-moore-skeleton} completes the first stage: for
every failed \(Q\), the erased ordinary-edge columns span a space of dimension
\(\rho m-\binom\rho2\), exactly their number.  It remains to place this
skeleton in a check space of dimension \(r_{\rm Sing}(n,\rho)\) and
to restore the loop columns.  Since
\[
 r_{\rm Sing}(n,\rho)
 =\rho m-\left(\binom\rho2-\rho\right),
\]
the ambient space \(F^\rho\) must lose
\(\binom\rho2-\rho\) binary dimensions.  At \(\rho=2\) this number is \(-1\),
so the quotient step applies only for \(\rho\geq3\).  For \(3\leq\rho\leq n\), set
\[
 g_\rho\triangleq\binom\rho2-\rho=\frac{\rho(\rho-3)}2,
\]
and, for \(Q\in\binom{\cA}{\rho}\), put
\[
 E_Q\triangleq\sum_{a\in Q}U_a^{(\rho)}\leq F^\rho.
\]
Equivalently, \(E_Q\) is the column space of
\(H_{\rm edge}^{(\rho)}|_{\Stedge(Q)}\).
Two tasks remain: find one \(g_\rho\)-dimensional quotient
kernel that misses every \(E_Q\), and then choose fixed loop columns that
complete each projected edge space.  The next proposition shows that, once
the first task has succeeded, the second is exactly the remaining condition
for Singleton optimality.

\begin{proposition}[Common-quotient and loop-completion criterion]
\label{prop:loop-completion}
Let \(3\leq\rho\leq n\).  Let \(Z\leq F^\rho\) have dimension \(g_\rho\) and satisfy
\[
 Z\cap E_Q=\{0\}
 \qquad\text{for every }Q\in\binom{\cA}{\rho}.
\]
Let
\[
 \pi_Z:F^\rho\longrightarrow F^\rho/Z
\]
be the quotient map.  For each \(a\in\cA\), choose
\(\ell_a\in F^\rho\) as the parity-check column assigned to the self-loop at
the vertex labeled \(a\).  Under the vertex relabeling
\(V_n\leftrightarrow\cA\), define the binary linear parity-check map
\[
 H_{Z,\ell}:\cG_2(n)\longrightarrow F^\rho/Z
\]
by assigning \(\pi_Z(h_{a,b}^{(\rho)})\) to the ordinary edge
\(\{a,b\}\) and \(\pi_Z(\ell_a)\) to the loop at \(a\), and set
\[
 \cC_{Z,\ell}\triangleq\ker H_{Z,\ell}.
\]
Then \(\cC_{Z,\ell}\) is a Singleton-optimal binary
\(\rho\)-node-erasure-correcting code if and only if, for every \(Q\), the
\(\rho\) loop images indexed by \(Q\) form a basis of
\[
 \frac{F^\rho/Z}{(E_Q+Z)/Z}.
\]
\end{proposition}

\begin{proof}
The projected check space has dimension
\[
 \rho m-g_\rho=\rho n-\binom\rho2,
\]
the graph Singleton redundancy.  Theorem~\ref{thm:uniform-moore-skeleton} gives
independent edge columns spanning \(E_Q\).  Because the kernel of
\(\pi_Z|_{E_Q}\) is \(Z\cap E_Q=\{0\}\), this projection preserves their full
rank for every \(Q\).  The projected edge space \((E_Q+Z)/Z\) has codimension
exactly \(\rho\) in the projected check space, so the erased columns are
independent precisely under the stated
loop-basis condition.  Their number equals the dimension of the projected
check space, so they then also span it; hence \(H_{Z,\ell}\) is onto and
\(r(\cC_{Z,\ell})=r_{\rm Sing}(n,\rho)\).
\end{proof}

This paper constructs no such \(Z\) for \(\rho\geq4\). For \(\rho=3\), one has \(g_3=0\).  The ambient space \(F^3\) already has
the required binary dimension \(3m=3n-3\), so \(Z=\{0\}\) and the
projection leaves \(F^3\) unchanged.  For every failed triple \(Q\), the erased-edge block has rank
\(3m-3\); the three loop columns at \(Q\) must therefore form a basis of the remaining
three-dimensional quotient.  For the edge skeleton, closing the one-bit gap is
this loop-completion problem.  We now suppress the superscript \((3)\), writing \(h_{a,b}\) and
\(U_a\).

The loop \(\ell_a\) is tested only for failed triples \(Q\) containing \(a\),
and then \(U_a\subseteq E_Q\).  Adding an element of \(U_a\) to \(\ell_a\)
therefore does not change its quotient class.  Since the first coordinate of
the map defining \(U_a\) is \(x\), we may normalize the first coordinate of
every loop to zero.  For the finite completion search, we now specialize the
arbitrary affine frame used above to zero and the Frobenius orbit of a normal element
\(\theta\in F\), whose Frobenius conjugates form a basis, namely the frame
\[
 \cA\triangleq\{0,\theta,\theta^2,\theta^{2^2},\ldots,\theta^{2^{m-1}}\}.
\]
We impose the Frobenius-symmetric choice below: coordinatewise squaring,
which sends each edge column \(h_{a,b}\) to \(h_{a^2,b^2}\), must also send each
loop column \(\ell_a\) to \(\ell_{a^2}\).  This determines the \(n-1\) loop
columns at the nonzero frame points from only two field elements, while the
loop at zero, which squaring must fix, requires two additional bits.  Thus choose \(y,z\in F\) and \(\epsilon,\delta\in\F_2\), and use
\begin{align}
 \ell_0&\triangleq(0,\epsilon,\delta),\label{eq:loopzero}\\
 \ell_{\theta^{2^i}}&\triangleq(0,y^{2^i},z^{2^i}),
 \qquad 0\leq i<m.
\label{eq:Frob-loops}
\end{align}

\begin{proposition}[Exact finite completions]
\label{prop:finite-completions}
For each
\[
 m\in\{5,7,9,11\},
 \qquad
 n=m+1\in\{6,8,10,12\},
\]
there exist a normal element \(\theta\in F\), field elements \(y,z\in F\),
and bits \(\epsilon,\delta\in\F_2\) for which the loop assignment
\eqref{eq:loopzero}--\eqref{eq:Frob-loops} completes the edge skeleton.
Consequently, there are binary optimal triple-node-erasure-correcting codes
of redundancy \(3n-3\) at these four lengths.
\end{proposition}

One set of polynomial-basis certificates is shown in
Table~\ref{tab:certificates}.  In each row, let
\(p_m(T)\in\F_2[T]\) be the polynomial displayed in the second column, take
\[
 F\triangleq\F_2[T]/(p_m(T))
\]
and let \(\alpha\) be the residue class of \(T\).  The entries
\(\theta,y,z\) are written in the polynomial basis
\(1,\alpha,\ldots,\alpha^{m-1}\), and \((\epsilon,\delta)\) is binary.
A row specifies the displayed construction.  For each row, an exact
computation forms the erased columns of every failed triple from these
parameters and checks the required rank condition.

\begin{table*}[t]
\centering
\caption{Exact loop-completion certificates in polynomial form.}
\label{tab:certificates}
\small
\setlength{\tabcolsep}{3pt}
\renewcommand{\arraystretch}{1.08}
\begin{tabular}{clllll}
\toprule
\(m\)&\(p_m(T)\)&\(\theta\)&\(y\)&\(z\)&\((\epsilon,\delta)\)\\
\midrule
5  &\(T^5+T^2+1\)&\(\alpha+1\)&\(\alpha^4+\alpha^2+\alpha+1\)&\(\alpha^4+\alpha^2+\alpha+1\)&\((1,0)\)\\
7  &\(T^7+T+1\)&\(\alpha^3+1\)&\(\alpha^5+\alpha^2+\alpha\)&\(\alpha^5+1\)&\((1,1)\)\\
9  &\(T^9+T+1\)&\(\alpha^5+\alpha^4+1\)&\(\alpha^7+\alpha^4+\alpha^3+\alpha^2+1\)&\(\alpha^6+\alpha^5+\alpha^3\)&\((1,1)\)\\
11 &\(T^{11}+T^2+1\)&\(\alpha^7+\alpha^4+\alpha^3+1\)&\(\alpha^8+\alpha^5+\alpha^3\)&\(\alpha^9+\alpha^8+\alpha^5+\alpha^3+1\)&\((1,1)\)\\
\bottomrule
\end{tabular}
\end{table*}

\begin{proof}
For each row of Table~\ref{tab:certificates}, exhaustive exact division by
every monic binary polynomial of degree between \(1\) and
\(\lfloor m/2\rfloor\) finds no divisor of the displayed degree-\(m\)
modulus; hence it is irreducible.  Binary row reduction verifies that
\[
 \theta,\theta^2,\ldots,\theta^{2^{m-1}}
\]
is an \(\F_2\)-basis of \(F\).  Hence the displayed labels form the required
affine frame.  For every \(Q\in\binom{\cA}{3}\), the exact verification
constructs the \(3m-3\) ordinary-edge columns
incident with \(Q\) and the three corresponding loop columns from
\eqref{eq:loopzero}--\eqref{eq:Frob-loops}.  Binary row reduction gives rank
\(3m-3\) before the loops are adjoined and rank \(3m\) afterward.  This is
checked for all
\[
 \binom63+\binom83+\binom{10}3+\binom{12}3=416
\]
failed triples.  Thus every erased triple indexes \(3m\) independent columns
in a check space of dimension \(3m=3n-3\).  The resulting code corrects every
triple and meets the graph Singleton bound.  The archived verification package
\cite{ZabokritskiyYohananov2026Verification} records all \(416\) exact rank
checks.
\end{proof}

\begin{remark}
The four completions establish Singleton optimality at
\(n=6,8,10,12\). By Proposition~\ref{prop:vertex-descent}, puncturing one
vertex from the \(n=10\) code also gives a Singleton-optimal binary code
correcting any two node erasures on \(9\) vertices, with redundancy \(17\).
The four completions lie within the restricted Frobenius-symmetric
choice \eqref{eq:loopzero}--\eqref{eq:Frob-loops}, for which this paper reports
certificates only at these four lengths; Proposition~\ref{prop:loop-completion} itself allows
any affine frame and arbitrary loop columns.  The uniform edge skeleton and the exact three-dimensional
completion criterion apply beyond these four instances.  Finding a uniform
solution of that criterion remains open; it would give an infinite family of
optimal triple-node codes.
\end{remark}

We finally summarize the division between the classical algebraic input and
the coding contribution of this section.  The Frobenius alternating-form layers in
\eqref{eq:Moore-column} come from the established theory of alternating
rank-metric spaces
\cite{DelsarteGoethals1975,GowQuinlan2009,Schmidt2010}.  For \(\rho=2\) and \(n\geq3\), each
space \(U_a^{(2)}=\{(x,a^2x+ax^2):x\in F\}\) is a member of the dimensional
dual hyperoval that Yoshiara associates with the Gold APN function
\(x\mapsto x^3\) \cite[Eqs.~(6) and~(7)]{Yoshiara2008}, and the case
\(\rho=2\) of Theorem~\ref{thm:uniform-moore-skeleton} is the property that two
distinct members meet in a one-dimensional subspace
\cite[Thm.~2.1(1) and Eq.~(9)]{Yoshiara2008}.  Kantor and Shult
\cite[Lem.~6.10]{KantorShult2013} prove the same count, \(\rho m-\binom\rho2\), for the
span of the multiples of \(\rho\le m\) independent linear forms among the quadratic
forms in \(m\) variables, whose space has dimension \(\binom{m+1}2=\binom n2\)
rather than \(\rho m\). Here the affine frame assigns these Frobenius
blocks to individual graph edges. For edges not incident with the
anchor, Lemma~\ref{lem:shifted-descent} rules out healthy endpoints and then forces
the strictly triangular matrix on the other failed vertices to vanish.
The first block then forces the coefficients of the edges at the anchor to vanish. This gives the local edge rank
in Theorem~\ref{thm:uniform-moore-skeleton} and leaves the loop-column
criterion of Proposition~\ref{prop:loop-completion}.

\section{Node-Weight Enumerators}
\label{sec:enumerators}

We now return to the \(q\)-ary model of Section~\ref{sec:preliminaries},
where \(q\) is an arbitrary prime power.

We use node-ball volumes to bound the size of a code at prescribed
distance and the number of centers needed to cover the space. At large
radii a separate Ramsey argument will determine when \(q^2\) centers
suffice. For the volumes themselves there are two counting routes.
Conditioning on the nonzero loops expresses
the ambient node-weight enumerator through a simple-graph polynomial.  Viewing the ball about zero
as the union of the star subspaces \(\cV_S\) instead gives a direct
inclusion--exclusion formula and Bonferroni estimates from its first two terms.  These
volumes become code bounds in Section~\ref{sec:bounds}.

Define the radius-\(t\) ball around zero by
\[
 \cB_t\triangleq
 \{G\in\cG_q(n):\wtN(G)\leq t\},
 \qquad 0\leq t\leq n,
\]
and put \(\cB_{-1}\triangleq\varnothing\).  Translation invariance implies
that every radius-\(t\) ball is a translate of this one and has the same
cardinality.

For a code \(\cC\subseteq\cG_q(n)\), define its node-weight distribution and
node-weight enumerator by
\begin{align}
 A_w^{\rm N}(\cC)
 &\triangleq
 |\{G\in\cC:\wtN(G)=w\}|,\label{eq:code-Aw}\\
 W_{\cC}^{\rm N}(z)
 &\triangleq
 \sum_{w=0}^n A_w^{\rm N}(\cC)z^w.
\label{eq:code-W}
\end{align}
For the entire ambient space, abbreviate
\[
 A_w^{\rm N}(n,q)\triangleq A_w^{\rm N}(\cG_q(n)),
 \qquad
 W_{n,q}^{\rm N}(z)\triangleq W_{\cG_q(n)}^{\rm N}(z).
\]
Because the metric is translation invariant, the distance distribution of
ordered pairs is determined by the ambient node-weight distribution:
\begin{equation}
 |\{(G,G')\in\cG_q(n)^2:d_{\rm N}(G,G')=w\}|
 =
 q^N A_w^{\rm N}(n,q).
\label{eq:pair-enumerator}
\end{equation}

\subsection{A loop-conditioning transform}

To count words in the looped graph space, we condition on the set of vertices
carrying nonzero loops.  Every vertex cover must contain these vertices.
After counting their loop labels and all ordinary-edge coordinates incident
with them, we are left with a simple support graph on the remaining vertices.
The vertex-cover enumerator of these simple graphs is the simple-graph
polynomial defined below.

For a simple graph \(J\) on the vertex set \([m]\), let \(E(J)\) denote
its edge set and let \(\tau(J)\) denote its vertex-cover number.  Write
\[
 K_m\triangleq\bigl([m],\binom{[m]}2\bigr)
\]
for the loopless complete graph on this vertex set.  Let
\begin{equation}
 s_{m,t}^{(q)}
 \triangleq
 \sum_{\substack{J\subseteq K_m\\\tau(J)=t}}
 (q-1)^{|E(J)|},
\label{eq:smt}
\end{equation}
where the sum is over all such support graphs.  We define
\(s_{m,t}^{(q)}\triangleq0\) when \(t<0\) or \(t>m\).
Define
\[
 S_{m,q}(z)\triangleq\sum_{t=0}^m s_{m,t}^{(q)}z^t,
\qquad
 S_{0,q}(z)\triangleq1.
\]

\begin{theorem}[Loop-conditioning transform]
\label{thm:loop-transform}
Let \(A_{w,\ell}^{\rm N}(n,q)\) count the \(q\)-ary looped graph words of
node weight \(w\) having exactly \(\ell\) nonzero loops.  Then
\begin{equation}
 A_{w,\ell}^{\rm N}(n,q)
 =
 \binom n\ell
 (q-1)^\ell
 q^{\,\ell n-\binom{\ell+1}{2}}
 s_{n-\ell,w-\ell}^{(q)}.
\label{eq:Aw-loop}
\end{equation}
Define the corresponding bivariate node-weight enumerator by
\begin{equation}
 W_{n,q}^{\rm N}(u,z)
 \triangleq
 \sum_{w=0}^n\sum_{\ell=0}^n
 A_{w,\ell}^{\rm N}(n,q)u^\ell z^w.
\label{eq:bivariate-node-enumerator}
\end{equation}
Consequently,
\begin{equation}
 W_{n,q}^{\rm N}(u,z)
 =
 \sum_{\ell=0}^n
 \binom n\ell
 (q-1)^\ell
 q^{\,\ell n-\binom{\ell+1}{2}}
 (uz)^\ell S_{n-\ell,q}(z).
\label{eq:loop-transform}
\end{equation}
In particular,
\[
 W_{n,q}^{\rm N}(z)=W_{n,q}^{\rm N}(1,z).
\]
\end{theorem}

\begin{proof}
Let \(L\) be the set of vertices carrying nonzero loops, with
\(|L|=\ell\).  Every vertex cover contains \(L\).  Once those vertices are
chosen, every ordinary edge incident with \(L\) is already covered and its
label is arbitrary.  The only remaining contribution to the node weight is
the vertex-cover number of the simple support graph induced on
\([n]\setminus L\).  There are
\[
 \ell(n-\ell)+\binom\ell2
 =
 \ell n-\binom{\ell+1}{2}
\]
ordinary-edge coordinates incident with \(L\).  Choosing \(L\), assigning its
nonzero loop labels, assigning these arbitrary incident labels, and then
choosing the residual graph gives \eqref{eq:Aw-loop}.  Summing over \(w\)
and \(\ell\) gives \eqref{eq:loop-transform}.
\end{proof}

The transform partitions the words by their sets of nonzero loops rather
than using inclusion--exclusion.
For example,
\[
 A_0^{\rm N}(n,q)=1,
\qquad
 A_n^{\rm N}(n,q)=(q-1)^nq^{\binom n2}.
\]
The simple-graph polynomial \(S_{m,q}\) carries the remaining difficulty of
the enumeration problem: it counts simple graphs by vertex-cover number, or
equivalently by independence number.

\subsection{Balls as a subspace arrangement}

For ball volumes there is a second representation tied directly to node
erasures: a word has node weight at most \(t\) exactly when it lies
in the coordinate subspace erased by some \(t\)-set of vertices.  This turns
the ball into the following finite union of linear subspaces.

Recall that \(\cV_S\) is the subspace of graph words supported on the erased
star coordinates \(\St(S)\).  Since every cover of size at most \(t\) extends to a
\(t\)-set,
\begin{equation}
\cB_t
=
\bigcup_{S\in\binom{[n]}t}\cV_S.
\label{eq:ball-union}
\end{equation}
The fixed-weight classes and balls determine each other by finite
difference.  In generating-function form,
\begin{equation}
 W_{n,q}^{\rm N}(z)
 =
 (1-z)\sum_{t=0}^{n-1}|\cB_t|z^t+q^Nz^n.
\label{eq:ball-to-weight}
\end{equation}
Indeed, \(A_t^{\rm N}(n,q)=|\cB_t|-|\cB_{t-1}|\) and
\(|\cB_n|=q^N\).  Thus every exact ball formula below immediately gives the
node-weight counts \(A_w^{\rm N}(n,q)\) exactly.

Inclusion--exclusion for \eqref{eq:ball-union} requires the dimension of every
intersection \(\bigcap_i\cV_{S_i}\).  That dimension depends only on which of
the cover sets contain each vertex, so it is determined
by a Venn profile rather than by the sets themselves.  For an
ordered \(j\)-tuple of \(t\)-sets \(S_1,\ldots,S_j\), put
\[
 I_j\triangleq\{1,\ldots,j\},
\]
and give each vertex index \(r\in[n]\) the membership signature
\[
 \sigma(r)\triangleq\{i\in I_j:r\in S_i\}.
\]
For \(T\subseteq I_j\), let
\[
 m_T\triangleq|\{r\in[n]:\sigma(r)=T\}|.
\]
The vector \(\boldsymbol m=(m_T)_{T\subseteq I_j}\) is the Venn profile of
the ordered family.

\begin{lemma}[Venn-profile intersection dimension]
\label{lem:profile-rank}
For a family with profile \(\boldsymbol m\),
\begin{equation}
 \dim\left(\bigcap_{i=1}^j\cV_{S_i}\right)
 =
 r_j(\boldsymbol m)
 \triangleq
 \frac12\left(
 \sum_{\substack{A,B\subseteq I_j\\A\cup B=I_j}}m_Am_B
 +m_{I_j}
 \right).
\label{eq:profile-rank}
\end{equation}
\end{lemma}

\begin{proof}
A loop coordinate at \(v\) belongs to every \(\St(S_i)\) exactly when
\(\sigma(v)=I_j\).  An ordinary edge \(\langle u,v\rangle\) belongs
to every \(\St(S_i)\) exactly when
\[
 \sigma(u)\cup\sigma(v)=I_j.
\]
The sum in \eqref{eq:profile-rank} counts ordered pairs satisfying this
condition.  Every off-diagonal pair occurs together with its reversal.  The
only diagonal pairs are the \(m_{I_j}\) vertices with full signature; each
corresponds to an allowed loop. Subtract the \(m_{I_j}\) diagonal pairs
and divide the remaining count by two to count ordinary edges.
Adding the \(m_{I_j}\) loops gives
\eqref{eq:profile-rank}.
\end{proof}

Let \(\mathfrak P_{n,t,j}\) be the set of all nonnegative integer vectors
\(\boldsymbol m=(m_T)_{T\subseteq I_j}\) satisfying
\begin{align}
 \sum_{T\subseteq I_j}m_T&=n,\label{eq:profile-total}\\
 \sum_{T\ni i}m_T&=t
 \qquad(i\in I_j),\label{eq:profile-rows}
\end{align}
and, for every distinct \(i,h\in I_j\),
\begin{equation}
 \sum_{\substack{T\subseteq I_j\\
 |\{i,h\}\cap T|=1}}m_T>0.
\label{eq:profile-distinct}
\end{equation}
Condition \eqref{eq:profile-distinct} says that the \(j\) cover sets are
pairwise distinct.

\begin{theorem}[Exact Venn-profile formula]
\label{thm:profile-IE}
For \(0\leq t\leq n\),
\begin{equation}
 |\cB_t|
 =
 \sum_{j=1}^{\binom nt}
 \frac{(-1)^{j+1}}{j!}
 \sum_{\boldsymbol m\in\mathfrak P_{n,t,j}}
 \frac{n!}{\prod_{T\subseteq I_j}m_T!}
 q^{\,r_j(\boldsymbol m)}.
\label{eq:profile-IE}
\end{equation}
\end{theorem}

\begin{proof}
Apply inclusion--exclusion to \eqref{eq:ball-union}.  Replace every
unordered \(j\)-element subfamily by its \(j!\) orderings.  For a fixed Venn
profile, the number of assignments of the \(n\) vertices to the
\(2^j\) signatures is the multinomial coefficient
\[
 \frac{n!}{\prod_Tm_T!}.
\]
Equations \eqref{eq:profile-total}--\eqref{eq:profile-rows} fix the
number of vertices and the set sizes, while \eqref{eq:profile-distinct} removes repeated
cover sets.  Hence each unordered subfamily appears under exactly its
\(j!\) orderings and is counted once after the factor \(1/j!\).
Lemma~\ref{lem:profile-rank} gives the size \(q^{r_j(\boldsymbol m)}\) of each intersection.
\end{proof}

Formula \eqref{eq:profile-IE} is exact, but its full upper limit remains
large.  Its algorithmic value is that every fixed inclusion--exclusion
order \(j\) uses only \(2^j\) profile variables and no longer enumerates
families of cover sets.

\subsection{Second-order formula and asymptotics}

The packing and existence results below need tractable estimates, and
for fixed-radius asymptotics the first two orders already determine the
leading term.  We therefore compute pair intersections explicitly and apply
the second Bonferroni inequality.

If \(S,T\in\binom{[n]}t\) and \(|S\cap T|=u\), then
\begin{equation}
 \dim(\cV_S\cap\cV_T)
 =
 un-\binom u2+(t-u)^2.
\label{eq:pair-intersection}
\end{equation}
The last term counts the edges between \(S\setminus T\) and
\(T\setminus S\); these cross edges meet both covers even
though neither endpoint lies in \(S\cap T\).

Here \(M_1\) will be the sum of the sizes of the individual star subspaces,
whereas \(M_2\) will be the sum of all pair-intersection sizes.  Define
\begin{align}
 M_1(n,t,q)
 &\triangleq
 \binom nt q^{\,tn-\binom t2},
\label{eq:M1}\\
 M_2(n,t,q)
 &\triangleq
 \frac12\binom nt
 \sum_{u=\max\{0,2t-n\}}^{t-1}\!
 \binom tu\binom{n-t}{t-u}
 \notag\\[-2pt]
 &\hspace{16mm}\cdot
 q^{\,un-\binom u2+(t-u)^2}.
\label{eq:M2}
\end{align}

\begin{theorem}[Pair-intersection bounds and fixed-radius asymptotics]
\label{thm:Bonferroni}
For every \(n\), prime power \(q\), and \(0\le t\le n\),
\begin{equation}
 M_1(n,t,q)-M_2(n,t,q)
 \leq|\cB_t|
 \leq M_1(n,t,q).
\label{eq:Bonferroni}
\end{equation}
Moreover, for fixed \(q\) and \(t\), as \(n\to\infty\),
\begin{align}
 |\cB_t|
 &=
 \binom ntq^{\,tn-\binom t2}
 \left(1+O_{q,t}(nq^{-n})\right),
\label{eq:ball-asymp}\\
 \log_q|\cB_t|
 &=
 tn-\binom t2+\log_q\binom nt+o(1).
\label{eq:ball-log}
\end{align}
The same leading asymptotic holds for \(A_t^{\rm N}(n,q)\).
\end{theorem}

\begin{proof}
The upper bound is the union bound in \eqref{eq:ball-union}; the lower bound
is the second Bonferroni inequality.  Equation
\eqref{eq:pair-intersection} and the number
\[
 \frac12\binom nt\binom tu\binom{n-t}{t-u}
\]
of unordered pairs with intersection \(u\) give \eqref{eq:M2}.
The asymptotic assertion is immediate when \(t=0\), so assume \(t\ge1\).

Writing \(s\triangleq t-u\), division by \(M_1\) gives
\begin{equation}
 \frac{M_2}{M_1}
 =
 \frac12
 \sum_{s=1}^{\min\{t,n-t\}}
 \binom ts\binom{n-t}s
 q^{-s(n-t)+\binom s2}.
\label{eq:M2-ratio}
\end{equation}
For fixed \(q,t\), the \(s=1\) term is
\(O_{q,t}(nq^{-n})\), and the remaining terms are smaller.  This proves
\eqref{eq:ball-asymp} and \eqref{eq:ball-log}: indeed,
\eqref{eq:Bonferroni} gives
\(0\le M_1-|\cB_t|\le M_2\), so the relative error is at most
\(M_2/M_1\).  Finally,
\[
 A_t^{\rm N}(n,q)=|\cB_t|-|\cB_{t-1}|,
\]
and the second term is exponentially smaller for fixed \(t\).
\end{proof}

The pair count also applies when the radius grows with \(n\). To see the
connection with the usual random-graph clique count, choose a graph word
uniformly and retain the vertices with loop label zero. Their number has
distribution \(\operatorname{Bin}(n,1/q)\). Join two retained vertices
if their edge label is zero; conditionally on their number this is an
Erd\H{o}s--R\'enyi graph with edge probability \(1/q\). The word belongs
to \(\cB_{n-m}\) exactly when this graph contains an \(m\)-clique.
For an ordinary random graph of fixed order, Matula's estimate
\cite[Thm.~5]{Matula1976} compares this upper-tail probability with the
expected clique count; the estimate below, for the present random number of
vertices, is proved directly. Here the first term \(M_1\), divided by \(q^N\), is the
expected clique count with the loop condition included.

Fix \(q\). As \(n\to\infty\), uniformly over the radii \(t\) with
\(m\triangleq n-t\ge2\log_q n\), both \(|\cB_t|\) and \(A_t^{\rm N}(n,q)\)
are asymptotic to \(M_1(n,t,q)\). Indeed, assume \(t>0\). In \eqref{eq:M2-ratio},
\(\binom as\le(ea/s)^s\) bounds the \(s\)-th summand before the factor
\(1/2\) by \(h(s)^s\), where
\[
 h(s)\triangleq\frac{e^2tm}{s^2}q^{-m+(s-1)/2}.
\]
Since \(\log h(s)\) is convex for \(1\le s\le m\), its maximum is at
an endpoint. Here \(h(1)\le e^2nmq^{-m}=o(1)\) and
\(h(m)\le e^2(n/m)q^{-(m+1)/2}=o(1)\) uniformly. Summing the
geometric bound gives \(M_2/M_1=o(1)\). For the shell, the preceding
ball has size at most \(M_1(n,t-1,q)\), whose ratio to \(M_1(n,t,q)\)
is \(tq^{-(m+1)}/(m+1)=o(1)\). The case \(t=0\) is exact.

\begin{corollary}[Radius one]
\label{cor:B1}
For every \(n\) and prime power \(q\),
\begin{align}
 |\cB_1|
 &=
 1+n(q^n-1)-\binom n2(q-1),\label{eq:B1}\\
 A_1^{\rm N}(n,q)
 &=
 n(q^n-1)-\binom n2(q-1).
\label{eq:A1}
\end{align}
\end{corollary}

\begin{proof}
There are \(n\) one-star subspaces, each of size \(q^n\).  A word in two
distinct one-star subspaces is supported on their common edge; a nonzero
word supported on one ordinary edge is counted twice, while the zero word is counted
in every star.  Separating these two cases gives \eqref{eq:B1}.
\end{proof}

Table~\ref{tab:binary-enumerators} lists the ambient binary distributions
and their cumulative ball volumes for \(1\leq n\leq6\).

\begin{table*}[t]
\centering
\caption{Binary ambient node-weight distributions and cumulative ball
volumes, verified by exhaustive enumeration. Dashes indicate indices greater
than \(n\).}
\label{tab:binary-enumerators}
\small
\setlength{\tabcolsep}{4pt}
\renewcommand{\arraystretch}{1.15}
\begin{tabular}{crrrrrrr}
\toprule
& \multicolumn{7}{c}{Node-weight counts \(A_w^{\rm N}(n,2)\)}\\
\cmidrule(lr){2-8}
\(n\) & \(w=0\) & \(1\) & \(2\) & \(3\) & \(4\) & \(5\) & \(6\)\\
\midrule
1 & 1 & 1 & -- & -- & -- & -- & --\\
2 & 1 & 5 & 2 & -- & -- & -- & --\\
3 & 1 & 18 & 37 & 8 & -- & -- & --\\
4 & 1 & 54 & 424 & 481 & 64 & -- & --\\
5 & 1 & 145 & 3,610 & 16,387 & 11,601 & 1,024 & --\\
6 & 1 & 363 & 25,130 & 382,157 & 1,124,572 & 532,161 & 32,768\\
\midrule
& \multicolumn{7}{c}{Ball volumes \(|\cB_t|\)}\\
\cmidrule(lr){2-8}
\(n\) & \(t=0\) & \(1\) & \(2\) & \(3\) & \(4\) & \(5\) & \(6\)\\
\midrule
1 & 1 & 2 & -- & -- & -- & -- & --\\
2 & 1 & 6 & 8 & -- & -- & -- & --\\
3 & 1 & 19 & 56 & 64 & -- & -- & --\\
4 & 1 & 55 & 479 & 960 & 1,024 & -- & --\\
5 & 1 & 146 & 3,756 & 20,143 & 31,744 & 32,768 & --\\
6 & 1 & 364 & 25,494 & 407,651 & 1,532,223 & 2,064,384 & 2,097,152\\
\bottomrule
\end{tabular}
\end{table*}

For example, for
\[
 n=4,\qquad t=2,\qquad q=2,
\]
\(M_1=768\) and \(M_2=432\), so
the second-order lower bound is \(336\), while direct enumeration gives
\[
 |\cB_2|=479.
\]
Full
inclusion--exclusion gives
\[
768-432+192-60+12-1=479.
\]
The same archive \cite{ZabokritskiyYohananov2026Verification} reproduces the
complete table and this inclusion--exclusion calculation using exact
arithmetic.

\section{Bounds from Ball Volumes}
\label{sec:bounds}

Let \(A_q^{\rm N}(n,d)\) denote the largest cardinality of a not
necessarily linear code in the graph-word space \(\cG_q(n)\) having minimum node distance at least
\(d\); its subscript is the alphabet size, whereas the subscript of
\(A_w^{\rm N}(n,q)\) is a weight. Let
\[
 t\triangleq\left\lfloor\frac{d-1}{2}\right\rfloor.
\]
Because the node metric is translation invariant, the standard sphere-packing
and maximal-code arguments apply verbatim.  Substituting the node-ball volumes
computed above gives the following specialization.

\begin{theorem}[Packing and Gilbert bounds]
\label{thm:packing-gilbert}
For every \(n\), prime power \(q\), and \(1\le d\le n+1\),
\begin{equation}
 \frac{q^N}{|\cB_{d-1}|}
 \le A_q^{\rm N}(n,d)
 \le \frac{q^N}{|\cB_t|}.
\label{eq:packing-gilbert}
\end{equation}
Consequently, if \(\rho=d-1\), then
\begin{equation}
 A_q^{\rm N}(n,\rho+1)
 \ge
 \frac{q^{N-\rho n+\binom\rho2}}{\binom n\rho}.
\label{eq:explicit-gv}
\end{equation}
\end{theorem}

The upper bound comes from disjoint radius-\(t\) balls, and the lower bound
from the radius-\((d-1)\) balls around a maximal distance-\(d\) code.  The last
display also uses
\[
 |\cB_\rho|\leq\binom n\rho q^{\rho n-\binom\rho2},
\]
because every word of node weight at most \(\rho\) lies in a \(\rho\)-star
subspace.

For \(d\geq3\), the upper bound in \eqref{eq:packing-gilbert} never improves
on the graph Singleton bound.  The puncturing argument of
Section~\ref{sec:preliminaries} applies to arbitrary codes and gives
\(A_q^{\rm N}(n,d)\leq q^{\binom{n-d+2}{2}}\).  On the other hand, let
\(\rho=d-1\) and \(t=\lfloor\rho/2\rfloor\).  Then
\(|\cB_t|\leq\binom nt q^{tn-\binom t2}\), and
\(\binom nt\leq2^{n-1}\) by Pascal's rule.  Since \(2t\leq\rho\leq n\),
\(r_{\rm Sing}(n,\rho)-tn+\binom t2=\sum_{j=t}^{\rho-1}(n-j)\geq n-1\): the
summands are positive, the first is \(n-t\), and for \(t\geq2\) there are
\(\rho-t\geq2\) of them, the first two giving \(2n-2t-1\geq n-1\).  Hence
\(|\cB_t|\leq q^{r_{\rm Sing}(n,\rho)}\), that is,
\(q^N/|\cB_t|\geq q^{\binom{n-d+2}{2}}\).  The node-ball volumes therefore
matter mainly for the existence and covering bounds.

Loeliger's averaging lemma applies to an arbitrary additive error set
\cite[Lem.~2 and Thm.~1]{Loeliger1994}.  Since every node ball about zero is
invariant under multiplication by \(\F_q^\times\), averaging over projective
error points gives a slightly sharper finite-length specialization.  Its binary
simple-graph counterpart is the argument of
\cite[Prop.~5]{KoppartyPotukuchiSha2025}. For fixed \(q\) and
\(\rho=d-1\ge1\), the resulting upper bound on redundancy exceeds the
graph Singleton redundancy by at most \(\rho\log_q n+O_{q,\rho}(1)\).

For \(0\leq s\leq n\), put
\[
\Lambda_s\triangleq\frac{|\cB_s|-1}{q-1};
\]
this is the number of projective points represented by the nonzero words in
the scalar-invariant ball \(\cB_s\).  Because the ball is closed under
nonzero scalar multiplication, a linear code contains a nonzero word of
\(\cB_s\) if and only if it contains the unique projective point represented by
that word.  Passing to points therefore represents each such containment
exactly once.

\begin{theorem}[Projective ball-volume bound for linear codes]
\label{thm:random-linear}
For every \(n\), prime power \(q\), and \(1\le d\le n+1\), there is a
linear code \(\cC\subseteq\cG_q(n)\) of node distance at least \(d\).
For \(d=1\), it may have redundancy zero.  For \(d\geq2\), its redundancy may
be chosen to satisfy
\begin{equation}
\begin{split}
 r(\cC)
 &\leq
 \min\left\{r\in\{1,\ldots,N\}:
 \Lambda_{d-1}(q^{N-r}-1)<q^N-1\right\}\\
 &\leq
 \min\left\{N,
 \max\left\{1,\left\lceil\log_q\Lambda_{d-1}\right\rceil\right\}
 \right\}.
\end{split}
\label{eq:random-linear}
\end{equation}
In particular, for fixed \(\rho=d-1\geq1\),
\begin{align}
 r(\cC)
 &\le
 \min\!\left\{N,
 \max\!\left\{1,
 \left\lceil\log_q\!\left[
 \binom n\rho
 \frac{q^{\rho n-\binom\rho2}-1}{q-1}
 \right]\right\rceil\right\}\right\},
 \label{eq:linear-explicit}\\
 r(\cC)
 &\le r_{\rm Sing}(n,\rho)+\rho\log_q n+O_{q,\rho}(1).
\label{eq:linear-asymptotic}
\end{align}
\end{theorem}

\begin{proof}
The assertion for \(d=1\) is immediate.  Let \(d\geq2\), fix
\(1\leq r\leq N\), and choose uniformly an \((N-r)\)-dimensional subspace
\(\cC\) of \(\cG_q(n)\).  The nonzero words in \(\cB_{d-1}\) split into
\(\Lambda_{d-1}\) projective points.  Each fixed point is contained in \(\cC\)
with probability
\[
 \frac{q^{N-r}-1}{q^N-1}.
\]
The expected number of forbidden points in \(\cC\) is therefore less than one under
the first condition in \eqref{eq:random-linear}, so some \(\cC\) contains none.
For
\[
 r\triangleq\max\left\{1,\left\lceil\log_q\Lambda_{d-1}\right\rceil\right\}<N,
\]
this expectation is at most
\[
 \frac{q^r(q^{N-r}-1)}{q^N-1}
 =\frac{q^N-q^r}{q^N-1}<1;
\]
if this choice is at least \(N\), take the zero code instead.  Finally, the
union of the \(\rho\)-star subspaces gives
\[
 \Lambda_\rho
 \leq
 \binom n\rho\frac{q^{\rho n-\binom\rho2}-1}{q-1}.
\]
This proves the explicit bound, and the fixed-\(\rho\) estimate of the binomial
coefficient gives the asymptotic form.
\end{proof}

The projective refinement can be strict.  For example, when \(q=3\), \(n=2\),
and \(d=2\), one has \(|\cB_1|=15\) and \(\Lambda_1=7\), so
\eqref{eq:random-linear} gives redundancy at most \(2\), whereas the
same bound with \(|\cB_1|\) in place of \(\Lambda_1\), namely
\(\lceil\log_3|\cB_1|\rceil\), gives \(3\).

For \(q=2\), fixed \(\rho\), and prime lengths \(n\) at which \(2\) is
primitive, the multi-slope codes of Theorem~\ref{thm:multislope} have
redundancy \(r_{\rm Sing}(n,\rho)+\binom{\rho-1}2\), which lies below the
guarantee \eqref{eq:random-linear} once \(n\) is large. Indeed, a redundancy \(r<N\) satisfying the condition in
\eqref{eq:random-linear} has \(2^r>\Lambda_\rho/2\), so by
\eqref{eq:ball-log} that guarantee is at least
\(r_{\rm Sing}(n,\rho)+\log_2\binom n\rho-O_\rho(1)\).  The existence bound
applies to all lengths and alphabets.

Let \(K_q^{\rm N}(n,R)\) be the minimum number of centers whose node balls
of radius \(R\) cover \(\cG_q(n)\).  The classical fractional-cover and greedy
set-cover bounds \cite{Lovasz1975,Chvatal1979} give a companion pair of
covering bounds.

\begin{proposition}[Covering bounds]
\label{prop:covering}
For every \(n\), prime power \(q\), and \(0\le R\le n\), put
\(b\triangleq|\cB_R|\) and \(H_b\triangleq\sum_{j=1}^b1/j\).  Then
\begin{equation}
 \frac{q^N}{b}
 \le K_q^{\rm N}(n,R)
 \le
 H_b\frac{q^N}{b}
 \le
 \frac{q^N}{b}\bigl(1+\ln b\bigr).
\label{eq:covering}
\end{equation}
\end{proposition}

\begin{proof}
Put \(v\triangleq q^N\) and recall \(b=|\cB_R|\).  In the incidence
hypergraph of translated balls, every block contains \(b\) words and every
word lies in \(b\) blocks.  Uniform weights \(1/b\) on the blocks and on the
points show, by the primal and dual fractional-cover programs, that the
fractional cover number is exactly \(v/b\).  By Chv\'atal's bound, the greedy
set cover has size at most \(H_b\) times the fractional optimum, proving the harmonic upper
bound; \(H_b\leq1+\ln b\) gives the last inequality.  The counting lower bound
is \(v/b\).
\end{proof}

Substituting \eqref{eq:ball-asymp} into
\eqref{eq:packing-gilbert}--\eqref{eq:covering} gives their asymptotic
forms for fixed \(q\), \(d\) and \(R\). For a ball of radius \(s\), the term
\[
 \log_q\binom ns
\]
in \eqref{eq:ball-log} counts, on the logarithmic scale, the possible
sets of \(s\) failed vertices. For fixed positive \(s\), the analogous
selection term gives the leading term in the logarithm of a Hamming-ball
volume; here it is a
lower-order correction to the dimension \(sn-\binom s2\). At \(\rho=d-1\),
the existence bound \eqref{eq:explicit-gv} is the graph Singleton
upper bound \(q^{\binom{n-\rho+1}{2}}\) divided by
\(\binom n\rho\).

At large radii a Ramsey number determines exactly when the minimum
covering size is \(q^2\). In the binary loopless model, Kopparty, Potukuchi, and Sha
\cite[p.~67:2]{KoppartyPotukuchiSha2025} relate distance from the empty
and complete graphs to the diagonal Ramsey problem. Here the loops also
have to match, and the centers may be arbitrary graph words.
Let \(R_q(m)\) be the least \(r\) such that every \(q\)-coloring of the
ordinary edges of the complete graph on \(r\) vertices contains a
monochromatic clique on \(m\) vertices.

\begin{theorem}[Covering at large radii]
\label{thm:ramsey-covering}
For \(2\leq m\leq n\),
\[
 K_q^{\rm N}(n,n-m)\geq q^2,\qquad
 K_q^{\rm N}(n,n-m)=q^2
 \quad\Longleftrightarrow\quad n>q\bigl(R_q(m)-1\bigr).
\]
When equality holds, the \(q^2\) words constant separately on loops
and on ordinary edges form a linear covering code.
\end{theorem}

\begin{proof}
A word is within distance \(n-m\) of a center if and only if they agree
on all coordinates of some \(m\)-vertex clique, including its loops.
First suppose there are fewer than \(q^2\) centers. At each vertex
choose a loop label used there by fewer than \(q\) centers. On each
ordinary edge choose a label avoided by all centers that match its
two chosen loops. Such a label exists, so no center agrees even on a pair.

Put \(r\triangleq R_q(m)\). If \(n>q(r-1)\), every graph word has at least \(r\)
vertices with the same loop label. Some \(m\) of them have one common
ordinary-edge label by the definition of \(r\). One of the stated
\(q^2\) centers therefore agrees on their clique.

Conversely, suppose \(n\leq q(r-1)\), and index any \(q^2\) centers
by a set \(I\). We choose a loop label at each vertex \(v\) and write
\(I_v\subseteq I\) for the centers matching it. Whenever possible choose
\(|I_v|<q\). At any remaining vertex the \(q\) available classes each
have size \(q\). Choose a class that has been selected at fewer than
\(r-1\) previous vertices. Such a choice exists: the \(q\) classes at
this vertex are distinct, and fewer than \(q(r-1)\) vertices have been
assigned. Thus no particular size-\(q\) class is chosen more than \(r-1\) times.

On an edge \(\langle u,v\rangle\) with \(|I_u\cap I_v|<q\), choose a label missed by
that intersection. The remaining edges lie within groups of vertices
having the same size-\(q\) class \(J\). Each group has at most \(r-1\)
vertices, so color its edges by the \(q\) center indices in \(J\)
with no monochromatic \(m\)-clique. If the centers in \(J\) miss a
label at an edge, choose that label. Otherwise they use every label
exactly once; choose the label of the center named by the edge color.
A center agreeing on an \(m\)-clique lies in \(I_u\cap I_v\) and matches the
label of each clique edge \(\langle u,v\rangle\), so each such intersection has
size \(q\); as \(m\geq2\), all clique vertices then share one size-\(q\) class,
and every clique edge has the color of that center, a monochromatic
\(m\)-clique inside one group. This contradiction proves necessity.
\end{proof}

For example, \(R_q(2)=2\), so the radius-\((n-2)\) covering number
equals \(q^2\) exactly when \(n\geq q+1\). At radius \(n-1\), agreement
on a single loop suffices: \(q\) centers realizing all loop labels at one
vertex cover the space, while for fewer centers some word differs from all of
them at every loop.

\section{The Complementary Clique Metric}
\label{sec:clique}

The coordinate complement of a union of failed-node stars is an induced
clique.  This elementary observation leads to a second metric whose balls
are substantially easier to enumerate and whose optimal codes are dual to
optimal node-erasure codes.

The full-rank condition on the complementary surviving clique was already
used in \cite[Thm.~21]{YohananovYaakobi2019} to derive bounds on optimal node-erasure
codes.  Here we study the erasure metric whose erasure patterns are these
complementary coordinate sets, give its weight distribution and ball volumes in closed form,
the latter recovering the inclusion--exclusion count of that theorem, and state
the optimal-code equivalence explicitly as a duality theorem.

Cliques also appear in the extremal theory of graph-codes recalled in the
Introduction.  A binary linear space of loopless graphs on the vertex set
\([n]\) contains no word equal to the edge set of a clique on two or more
vertices (with all other vertices isolated) if and only if, for every
set of at least two vertices, some graph in the dual space has an odd number
of edges inside that set; Alon \cite{Alon2024GraphCodes} proved that the
minimum codimension of such a space is \(\lfloor n/2\rfloor\) (see also
\cite{Alon2025GraphCodesSurvey}).  Linear codes of large clique distance,
defined below with the loops among the coordinates, satisfy a support-based
condition instead: no nonzero codeword is supported inside an induced clique
on at most a prescribed number of vertices; equivalently, the parity-check
columns indexed by each such clique are linearly independent.  Over
\(\F_2\), this is stronger, clique by clique, than excluding the loopless
clique itself, which needs only one parity check of odd weight on its
ordinary edges, but it does not constrain larger cliques.

For \(T\subseteq[n]\), recall from Section~\ref{sec:preliminaries} the induced
clique-coordinate set \(\Cl(T)\) of all coordinates with both endpoints in
\(T\), loops included.  For \(e\in E_n\), write \(x_e\) for the
coordinate of \(x\in\cG_q(n)\) indexed by \(e\).  Define the vertex support
and clique weight by
\begin{align}
 \operatorname{vsupp}(x)
 &\triangleq\{i\in[n]:x_{\langle v_i,v_j\rangle}\ne0
 \text{ for some }j\in[n]\},\label{eq:vsupp}\\
 \wtCl(x)
 &\triangleq\min\{|T|:\supp(x)\subseteq\Cl(T)\}
 =|\operatorname{vsupp}(x)|.
\label{eq:clique-weight}
\end{align}
A nonzero loop activates one vertex and a nonzero ordinary edge activates
both endpoints.  For example, a word supported on one ordinary edge has
node weight one but clique weight two.

\begin{proposition}
Let \(0\leq\mu\leq n\) be an integer.
The function
\[
 d_{\rm cl}:\cG_q(n)\times\cG_q(n)\longrightarrow\{0,1,\ldots,n\},
 \qquad
 d_{\rm cl}(x,y)\triangleq\wtCl(x-y)
\]
is a translation-invariant metric.  A code corrects every erasure of an
induced clique on at most \(\mu\) vertices if and only if its minimum clique
distance is at least \(\mu+1\).
\end{proposition}

\begin{proof}
The only nontrivial metric axiom follows from
\[
 \operatorname{vsupp}(x+y)
 \subseteq
 \operatorname{vsupp}(x)\cup\operatorname{vsupp}(y).
\]
The erasure assertion is the standard support criterion for a known
coordinate-erasure pattern.
\end{proof}

For an integer \(0\leq\mu\leq n\), if a linear code
\(\mathcal D\leq\cG_q(n)\) corrects every erasure of an induced clique on
\(\mu\) vertices, then its redundancy satisfies the clique Singleton bound
\begin{equation}
 r(\mathcal D)\geq |\Cl(T)|=\binom{\mu+1}{2},
 \qquad |T|=\mu.
\label{eq:clique-singleton}
\end{equation}
Indeed, apply Lemma~\ref{lem:erased-column-criterion} to the coordinate-erasure
set \(X=\Cl(T)\): its parity-check columns must be linearly independent.  A
code attaining equality in \eqref{eq:clique-singleton} is
called Singleton-optimal for \(\mu\)-clique erasures.

\subsection{Exact enumerator and generating function}

For \(0\leq t\leq n\), define
\[
 W_{n,q}^{\rm cl}(z)
 \triangleq
 \sum_{x\in\cG_q(n)}z^{\wtCl(x)},
 \qquad
 \cB_t^{\rm cl}
 \triangleq
 \{x\in\cG_q(n):\wtCl(x)\leq t\}.
\]

For simple graphs, counting labeled graphs with no isolated
vertices by inclusion--exclusion is classical; see
\cite{BenderCanfieldMcKay1997} for a refined enumeration by edge number.  The
expression below is the corresponding \(q\)-ary looped analogue.  For
every integer \(m\geq0\), let \(\nu_m(q)\) be the number of \(q\)-ary
graph words on a fixed set of
\(m\) vertices whose vertex support is the entire set; equivalently,
whose support graph has no isolated vertex.  Inclusion--exclusion
over the vertices absent from the support gives
\begin{equation}
\nu_m(q)
=
\sum_{j=0}^m(-1)^j\binom mj
q^{\binom{m-j+1}{2}}.
\label{eq:nu-m}
\end{equation}
Indeed, after \(j\) specified vertices are absent, the remaining looped
complete graph has \(m-j\) vertices and
\(\binom{m-j}{2}+(m-j)=\binom{m-j+1}{2}\) free coordinates.  This count
includes its loops; inclusion--exclusion enforces that no remaining vertex
is isolated.

\begin{theorem}[Clique-weight enumerator]
\label{thm:clique-enumerator}
For \(0\le t\le n\), the ambient clique-weight enumerator and ball volume
are
\begin{align}
 W_{n,q}^{\rm cl}(z)
 &=\sum_{m=0}^n\binom nm\nu_m(q)z^m,
 \label{eq:clique-enumerator}\\
 |\cB_t^{\rm cl}|
 &=\sum_{m=0}^t\binom nm\nu_m(q).
\label{eq:clique-ball}
\end{align}
For a formal indeterminate \(u\), the exponential generating function of the
polynomials \(W_{n,q}^{\rm cl}(z)\) satisfies the formal identity
\begin{equation}
 \sum_{n\ge0}W_{n,q}^{\rm cl}(z)\frac{u^n}{n!}
 =
 e^{(1-z)u}
 \sum_{m\ge0}q^{\binom{m+1}{2}}\frac{(zu)^m}{m!}.
\label{eq:clique-egf}
\end{equation}
\end{theorem}

\begin{proof}
Choose the \(m\) active vertices and then count graph words with precisely that
vertex support using \eqref{eq:nu-m}.  Summing up to \(t\) proves
\eqref{eq:clique-ball}.  For \eqref{eq:clique-egf}, substitute
\eqref{eq:nu-m}, set \(r\triangleq m-j\), and multiply exponential generating functions:
\begin{align*}
 \sum_{n\ge0}W_{n,q}^{\rm cl}(z)\frac{u^n}{n!}
 &=e^u\sum_{m\ge0}\nu_m(q)\frac{(zu)^m}{m!}\\
 &=e^ue^{-zu}
 \sum_{r\ge0}q^{\binom{r+1}{2}}\frac{(zu)^r}{r!}.
\end{align*}
For \(q>1\), the last series has radius of convergence zero; the asserted identity
is therefore an identity of formal power series.
\end{proof}

\subsection{Node--clique duality for optimal codes}

For \(S\subseteq[n]\), the star and complementary clique coordinate sets
form the disjoint partition
\begin{equation}
 \St(S)\mathbin{\dot\cup}\Cl([n]\setminus S)=E_n.
\label{eq:star-clique-partition}
\end{equation}
The dual below is taken with respect to the edge-coordinate inner product
\[
 \langle\cdot,\cdot\rangle_E:
 \cG_q(n)\times\cG_q(n)\longrightarrow\F_q,
 \qquad
 \langle x,y\rangle_E\triangleq\sum_{e\in E_n}x_ey_e.
\]
For a linear code \(\cC\leq\cG_q(n)\), define
\[
 \cC^\perp\triangleq
 \{y\in\cG_q(n):\langle x,y\rangle_E=0
   \text{ for every }x\in\cC\}.
\]
This convention matters in characteristic two: the Frobenius inner product
of full symmetric matrices counts every off-diagonal coordinate twice and
therefore pairs only the loop coordinates.

The information-set duality used below is standard; in an erasure-pattern
formulation it appears, for example, in
\cite[Sec.~V]{HolzbaurPuchingerYaakobiWachterZeh2021}.  The graph-specific
ingredient is the star--clique partition \eqref{eq:star-clique-partition}.

\begin{theorem}[Node--clique duality]
\label{thm:node-clique-duality}
Let \(\rho\) and \(\mu\) be nonnegative integers with \(\rho+\mu=n\), and
let \(\cC\le\cG_q(n)\) have dimension
\[
 \dim\cC=\binom{\mu+1}{2}.
\]
Then \(\cC\) is Singleton-optimal for correcting every \(\rho\)-node
erasure if and only if \(\cC^\perp\) is Singleton-optimal for correcting
every erasure of an induced \(\mu\)-vertex clique.
\end{theorem}

\begin{proof}
Fix \(S\subseteq[n]\) with \(|S|=\rho\), and put
\(T\triangleq[n]\setminus S\).
Choose a full-row-rank generator matrix
\(G_{\cC}\in\F_q^{\binom{\mu+1}{2}\times N}\) for \(\cC\), with columns
indexed by \(E_n\).
The code \(\cC\) corrects the coordinates \(\St(S)\) precisely when the
projection
\[
\pi_{\Cl(T)}:\cC\longrightarrow\F_q^{\Cl(T)}
\]
is injective: its kernel consists exactly of the codewords supported on
\(\St(S)\).  Its domain and codomain both have dimension
\(\binom{\mu+1}{2}\), so this is equivalent to the projection being an
isomorphism.  Equivalently, the columns of \(G_{\cC}\) indexed by
\(\Cl(T)\) are linearly independent.  The matrix \(G_{\cC}\) is a
parity-check matrix for \(\cC^\perp\); hence this is exactly the criterion
that \(\cC^\perp\) correct the erasure of \(\Cl(T)\).  Its redundancy is
\[
 N-\dim\cC^\perp=\dim\cC=|\Cl(T)|,
\]
so the clique-erasure code is Singleton-optimal.  The converse is
symmetric.
\end{proof}

\begin{corollary}
The duals of the optimal double-node-erasure codes of
\cite[Constr.~6, Thm.~12]{YohananovYaakobi2019} form an unconditional infinite family of
optimal binary \((n-2)\)-clique-erasure codes, of dimension \(2n-1\), at
every prime length \(n\geq5\).  The duals of the
optimal triple-node completions in Section~\ref{sec:optimal-triples} give
optimal \((n-3)\)-clique-erasure codes at \(n=6,8,10,12\).
\end{corollary}

The range \(n\geq5\) in the corollary is that of the cited construction.
The later double-node construction of
\cite[Constr.~1, Thm.~2]{YohananovEfronYaakobi2020} also works at every odd
prime length; Theorem~\ref{thm:multislope} includes that code at \(\rho=2\)
under the hypothesis, common to the multi-slope family, that \(2\) is
primitive modulo \(n\).

In the language of this section, the inclusion--exclusion count in
\cite[Thm.~21]{YohananovYaakobi2019} is the volume of the radius-\(t\) clique
ball computed in Theorem~\ref{thm:clique-enumerator}, and its inequality is
the sphere-packing bound in the clique metric for the dual of an optimal
\((n-2t)\)-node-erasure code.  By Theorem~\ref{thm:node-clique-duality}, this
dual has minimum clique distance at least \(2t+1\) and dimension
\(N-\binom{2t+1}2\), so the bound reads \(|\cB_t^{\rm cl}|\leq q^{\binom{2t+1}2}\).

\begin{remark}
Each side of the equivalence in Theorem~\ref{thm:node-clique-duality}
already forces \(\dim\cC=\binom{\mu+1}2\) through its Singleton equality, so
the theorem also holds without its dimension hypothesis.  Under that
hypothesis, the proof shows that \(\cC\) corrects every \(\rho\)-node erasure
if and only if \(\cC^\perp\) corrects every erasure of an induced
\(\mu\)-vertex clique.  Without it this correction-only equivalence fails
when \(\mu\geq1\): the zero
code corrects every coordinate erasure, whereas its dual is the ambient space
and corrects no nonempty erasure.  The theorem is a
direct dual reformulation of the full-rank criterion in
\cite[Lem.~17]{YohananovYaakobi2019}, which requires the generator-matrix
columns indexed by every \(\mu\)-vertex clique to be linearly independent.  The graph-specific metric formulation,
its closed-form enumerator for all clique weights, and the code-family applications above are the
contributions here.
\end{remark}

This duality concerns correctable erasure patterns of codes of dimension
\(\binom{\mu+1}2\).  For general linear codes, a node-weight enumerator alone does
not determine the node-weight enumerator of the dual.
Appendix~\ref{app:macwilliams} gives the precise obstruction and a
character-sum formula using finer information.

\section{Discussion and Open Problems}
\label{sec:discussion}

Allowing the cyclic slope to vary gives binary triple-node codes of
redundancy \(3n-2\) unconditionally at infinitely many prime lengths,
including lengths modulo which \(2\) is not primitive.  It remains unknown whether binary triple-node codes with redundancy
\(3n-3\) exist for infinitely many \(n\); such a family would close the
one-bit gap to the graph Singleton bound.

Three further directions emerge from the results.  First, the multi-slope
family, available at primes \(n\) modulo which \(2\) is primitive, is uniform
in the number of failed vertices and
improves on Schmidt's rank-metric construction for every
\(2\leq\rho<(n+1)/2\).  Its redundancy gap above the graph Singleton bound
\[
 \binom{\rho-1}{2}
\]
is independent of \(n\), but grows quadratically with the number of
failures. The space of dependencies among the \(\rho n\) check rows has dimension
\(\rho-1\). A basis consists of the coefficient vectors of the relations obtained by
summing the rows within the neighborhood family and within each of the \(\rho-2\)
slope families. Attaining the graph Singleton redundancy with \(\rho n\) rows would require dimension
\(\binom\rho2\). The diagonal family has no internal row dependency and every other
family has at most one. Any \(\rho\) families of \(n\) checks with these
properties can close that gap only through dependencies coupling different
families; the families used here have none.

Second, Theorem~\ref{thm:uniform-moore-skeleton} shows uniformly for every
\(2\leq\rho\leq n\) that
the erased-edge columns of the edge skeleton are independent for every failed set.  For
\(3\leq\rho\leq n\), the completion of Proposition~\ref{prop:loop-completion}
quotients the check space by a common subspace of dimension \(\rho(\rho-3)/2\)
and then chooses loop columns that complete every failed set; for
\(\rho\geq4\) no such subspace is exhibited here.
At \(\rho=3\) this subspace is zero.  For each failed triple, the span of the erased-edge columns
then has codimension three in the check space, and the three corresponding
loop columns must form a basis of the remaining three-dimensional quotient.  At the prime \(n=5\), the published slope-two code has no
one-dimensional extension, that is, no supercode of dimension one larger,
that still corrects every failed triple, yet an optimal
binary triple-node code is given by
\cite[Thm.~18]{YohananovYaakobi2019}.  The loop completions at \(n=6,8,10,12\) give further optimal
codes, at lengths where the cyclic construction is not defined.  An algebraic completion valid for infinitely many \(n\) would
give an infinite optimal family.

Applied
repeatedly, Proposition~\ref{prop:vertex-descent} turns any optimal family
for a fixed number \(\rho\geq4\) of failures into an optimal family for three
failures on \(\rho-3\) fewer vertices.

Third, the ball of radius \(t\) about zero is the union of the \(\binom nt\)
subspaces of words supported on the coordinates erased by \(t\) failed vertices. The exact
inclusion--exclusion formula counts their intersections, and its first two
terms show that at fixed radius the ball is asymptotic in size to the
sum of their sizes. At the opposite end, Theorem~\ref{thm:ramsey-covering}
determines exactly when \(q^2\) balls of radius \(n-m\) cover the space. It leaves
open the minimum number of centers when \(n\leq q(R_q(m)-1)\).
Other enumerative questions include a recurrence for the distributions of
vertex-cover numbers of simple graphs with nonzero edge labels, to which the ambient
node-weight distribution reduces, and efficient evaluation of the exact
formula when \(t\) grows with \(n\).

The complementary clique metric also leads to the node--clique duality,
which reformulates the full-rank criterion of \cite[Lem.~17]{YohananovYaakobi2019}:
Singleton-optimal node-erasure codes and clique-erasure codes are dual under the
edge-coordinate inner product.  In particular, existing optimal binary
double-node codes yield an unconditional infinite family of optimal
clique-erasure codes.

For the node weights of a dual code, Appendix~\ref{app:macwilliams}
shows that counts by support graph up to vertex relabeling suffice
whereas counts by node weight
alone do not. It also shows that, among the partitions of nonempty supports whose class counts
determine the node-weight distribution of the dual, the coarsest one groups the
supports of words whose additive-character sums over each node-weight class
agree,
as in the dual partition of \cite[Def.~2.1]{GluesingLuerssen2015}.
Can one describe these classes directly from the support graph,
without computing the character sums?

The recent asymptotic work of Kopparty, Potukuchi, and Sha
\cite{KoppartyPotukuchiSha2025} studies the node metric in the regime of
constant relative distance, which complements the present regime of fixed
distance.  Finding explicit
families that interpolate between the two regimes is a natural problem for later work.

\section*{Computer-Assisted Verifications}

The following finite statements were verified by exact computer
calculation: the four loop completions in
Proposition~\ref{prop:finite-completions} (irreducibility of the moduli,
normality of \(\theta\), and the rank conditions for all \(416\) failed
triples), the binary node-weight distributions and ball volumes in
Table~\ref{tab:binary-enumerators} together with the inclusion--exclusion
example following it, and the certificates in Appendix~\ref{app:one-bit} showing that, at
\(n=5,11,13,19\), the published slope-two code has no one-dimensional
extension that still corrects every triple of failed vertices.  The computational
certificates are archived in \cite{ZabokritskiyYohananov2026Verification}.

\begin{appendices}
\renewcommand{\theHsection}{appendix.\Alph{section}}
\renewcommand{\theHequation}{appendix.\Alph{section}.\arabic{equation}}
\renewcommand{\theHtheorem}{appendix.\Alph{section}.\arabic{theorem}}

\section{Proof of the Local Determinant Criterion}
\label{app:primitive-slope}

\begin{proof}[Proof of Lemma~\ref{lem:local-determinant-criterion}]
We first verify that the change of variables used after
Theorem~\ref{thm:primitive-slope} loses no words supported on the erased
stars. This step
works over any characteristic-two coefficient field and for any nonempty failed
set \(A\) at an odd prime \(n\). Choose \(a_0\in A\) and use the independent
variables \(Z_{a_0}\) and \(Y_v\triangleq Z_v+Z_{a_0}\) for \(v\ne a_0\); put \(Y_{a_0}\triangleq0\).
In these coordinates the coefficient of \(Z_{a_0}Y_v\) is the sum of the
ordinary-edge labels incident to \(v\), and the neighborhood sum at \(a_0\)
is the sum of those at \(v\ne a_0\). The coefficient of \(Z_{a_0}^2\) is the
sum of all labels. Thus the neighborhood checks and the total diagonal
check vanish exactly when \(P\) contains no \(Z_{a_0}\). A labeling is
supported on \(\St(A)\) precisely when setting \(Z_a=0\) for all \(a\in A\)
annihilates its formal polynomial. For a polynomial containing no \(Z_{a_0}\),
this is equivalent to each monomial having a factor \(Y_a\) with
\(a\in A\setminus\{a_0\}\).

Now take \(A=I=\{0,a,b\}\), \(a_0=0\), and work over \(\widetilde K\).
Let \(W\) be the \(\widetilde K\)-span of the \(n-1\) independent
forms \(Y_v\), \(v\ne0\). The map
\[
 W\longrightarrow\widetilde K^{\mathbb Z_n^\times},\qquad
 L\longmapsto\bigl(L(\mathbf t_k)\bigr)_{k\ne0}
\]
is an isomorphism. Indeed, \(L(\mathbf t_0)=0\) holds by definition of
the \(Y_v\); if all the other evaluations vanish, the polynomial of
degree at most \(n-1\) whose coefficients are those of \(L\) in the
\(Z_v\) vanishes at all \(n\) distinct roots of \(x^n-1\). Hence \(L=0\),
and the two spaces have the same dimension.

By the support calculation, \(P=Y_a L_a+Y_b L_b\) with \(L_a,L_b\in W\).
At \(k\ne0\) the diagonal equations are equivalent, as shown after
Theorem~\ref{thm:primitive-slope}, to
\[
 L_a(\mathbf t_k)=d_b(k)y_k,\qquad
 L_b(\mathbf t_k)=d_a(k)y_k,\qquad d_i(k)=1+\zeta^{ki}.
\]
Conversely, any tuple \((y_k)_{k\ne0}\) gives two such forms by the
isomorphism above. The resulting polynomial represents a supported
word satisfying all neighborhood and diagonal checks: the zero-frequency
diagonal check is already imposed, and the other diagonal evaluations
vanish by construction. This polynomial is zero exactly when
\((L_a,L_b)=c(Y_b,Y_a)\) for a scalar \(c\). Indeed,
\(Y_aL_a=Y_bL_b\) first forces \(L_a\) to be a multiple of the independent
variable \(Y_b\), and then determines \(L_b\). By the evaluation
isomorphism these pairs correspond precisely to the constant tuples.

It remains to impose the slope checks. The zero-frequency polar
evaluation vanishes automatically; at \(k\ne0\), \eqref{eq:lambda-factorization}
requires \(y_k=y_{\lambda k}\) exactly when
\(\Delta_k(a,b)\ne0\), and imposes no condition otherwise.
These evaluations test the entire slope family because its syndrome
polynomial has degree less than \(n\). Each orbit \(O\) of multiplication by
\(\lambda\) is a directed cycle with links \(k\to\lambda k\), and the link at
\(k\) is cut when \(\Delta_k(a,b)=0\). If \(z(O)=0\), the cycle contributes one
free value; otherwise the \(z(O)\) cut links split it into \(z(O)\) components,
each with one free value.
This also counts the two opposite links of a two-cycle correctly.
The one-dimensional space of globally constant tuples represents zero,
so the dimension of the kernel is \(\sum_O\max\{1,z(O)\}-1\).
The same nullity holds over \(\F_2\) because the original matrix is binary.

A nonprimitive slope has at least two orbits and hence positive nullity.
For a primitive slope there is one orbit. If \(\Delta_k(a,b)=0\), then
\(\Delta_{2k}(a,b)=\Delta_k(a,b)^2=0\). Hence the determinant vanishes on the whole coset \(k\langle2\rangle\),
of size \(\operatorname{ord}_n(2)\ge3\) (\(n\ge5\)), so a single
zero cannot occur. This proves both directions of the recovery criterion
and the asserted lower bound.
\end{proof}

\section{Rank of the Cyclic-Slope Check Matrix}
\label{app:primitive-slope-rank}

\begin{proof}[Proof of Proposition~\ref{prop:primitive-slope-rank}]
More generally, let \(n\) be any odd prime and take a possibly empty set \(\Lambda\subseteq
\mathbb Z_n^\times\setminus\{1,-1\}\) with
\(\Lambda\cap\Lambda^{-1}=\varnothing\), using one family of slope checks
for each of its elements. We show that these checks together with the
neighborhood and diagonal checks have rank
\[
 n+(|\Lambda|+1)(n-1).
\]
In a binary dependence, each diagonal row contains its own loop (since two is invertible modulo
\(n\)), and loops occur in no other rows. Thus
all diagonal coefficients vanish. If \(a\) and \(c_\lambda\) are the
remaining coefficient functions, the column of \(\{u,v\}\) gives
\[
 a(u)+a(v)+\sum_{\lambda\in\Lambda}
 \bigl(c_\lambda(u+\lambda v)+c_\lambda(v+\lambda u)\bigr)=0.
\]
This holds also at \(u=v\), by cancellation. The linear forms
\(u,v,u+\lambda v,v+\lambda u\) occurring here are pairwise
nonproportional over \(\mathbb Z_n\), by the hypotheses on \(\Lambda\).
Choose any one form \(L\), let \(f\) be the coefficient function applied to it, and sum the identity over its affine line
\(L(u,v)=t\). The term \(f(L)\) contributes \(nf(t)=f(t)\) in
\(\F_2\). Each other form takes every value once on that line,
so its contribution is independent of \(t\). Hence \(f\) is constant.
Conversely, summing all rows of the neighborhood family or of any one
slope family gives zero, since every ordinary edge occurs twice.
The coefficient vectors of these relations have disjoint supports, so
the space of dependencies among the \((2+|\Lambda|)n\) rows has dimension
exactly \(1+|\Lambda|\), proving the formula.
For \(\lambda\in\mathbb Z_n^\times\setminus\{\pm1\}\) one has
\(\lambda^{-1}\ne\lambda\), so \(\Lambda=\{\lambda\}\) satisfies the hypotheses; apply
the formula with \(\Lambda=\{\lambda\}\).
\end{proof}

\section{Proof of the Multi-Slope Theorem}
\label{app:multislope}

The diagonal evaluation supplies the product equation in the next lemma;
it takes the place of the top Frobenius equation in
Lemma~\ref{lem:shifted-descent}. In this case the coefficient matrix
must have zero diagonal as well as being symmetric: both properties
are needed to eliminate the squares and mixed terms of the formal
polynomial. This is the hyperbolic case of \cite[Thm.~2.4]{Quillen1971}.

\begin{lemma}[Frobenius descent]
\label{lem:frobenius-descent}
Let \(F\) be a field of characteristic two, and let \(s\geq1\) be an integer.
Let
\(\xi_1,\ldots,\xi_s\in F\) be linearly independent over \(\F_2\), and
suppose that \(\eta_1,\ldots,\eta_s\in F\) satisfy
\begin{align}
 \sum_{i=1}^s \xi_i\eta_i&=0,
 \label{eq:descent-product}\\
 \sum_{i=1}^s
 \left(\xi_i\eta_i^{2^k}+\xi_i^{2^k}\eta_i\right)&=0,
 &&1\le k\le s-1.
\label{eq:descent-moore}
\end{align}
Then there is a symmetric binary matrix \(C\) of size \(s\) with zero
diagonal such that \(\eta_i=\sum_{j=1}^s C_{ji}\xi_j\) for every \(i\).
\end{lemma}

\begin{proof}
Use the matrices \(X,Y,C=X^{-1}Y\) from the proof of
Lemma~\ref{lem:shifted-descent}: \(X\) and \(Y\) have the rows
\(x_j=(\xi_1^{2^j},\ldots,\xi_s^{2^j})\) and
\(y_j=(\eta_1^{2^j},\ldots,\eta_s^{2^j})\) for \(0\leq j<s\), and \(C^{[2]}\) is
the entrywise square of \(C\). Equations~\eqref{eq:descent-moore} are
precisely the first \(s-1\) equations used there. Hence \(C\) is symmetric
and, by \eqref{eq:descent-common-rows}, the \(s-1\) independent rows
\(x_1,\ldots,x_{s-1}\) annihilate \(C+C^{[2]}\), which therefore has rank at
most one.
We now use the product equation in place of the final Frobenius equation.

With \(x_j,y_j\) as above, \(x_jy_j^{\mathsf T}=0\) for
\(0\leq j<s\), by raising \eqref{eq:descent-product} to the \(2^j\)-th
power. Since \(y_j=x_jC\) and \(C\) is symmetric, this says that
\(\sum_i C_{ii}\xi_i^{2^{j+1}}=0\) for each \(j\). The coefficient matrix
is the entrywise square of \(X\), hence invertible; thus \(C\) has zero
diagonal. So does the symmetric matrix \(C+C^{[2]}\). If an off-diagonal
entry were nonzero, the corresponding principal \(2\)-by-\(2\) minor
would be nonzero, contrary to the rank bound. Hence \(C=C^{[2]}\) is
binary. The required coefficients are those in \(y_0=x_0C\).
\end{proof}

\begin{proof}[Proof of the correction assertion in Theorem~\ref{thm:multislope}]
It suffices to consider a failed set \(A\) of size \(\rho\), since a
smaller one can be enlarged. Use the binary linear forms \(Y_v,L_i\), for
which \(P=\sum_{i=1}^sY_{a_i}L_i\), and
the field \(F=\F_2(\zeta)\) defined after Theorem~\ref{thm:multislope}.
Evaluation at \(\mathbf t=(\zeta^v)_{v\in\mathbb Z_n}\) maps the binary
span of the \(n-1\) forms \(Y_v\), \(v\ne a_0\), isomorphically onto
\(F\). Indeed a binary combination of the \(Y_v(\mathbf t)\) gives an
even-parity coefficient vector on \(1,\zeta,\ldots,\zeta^{n-1}\).
The only nonzero relation on these powers has all \(n\) coefficients
equal to one: \(1+T+\cdots+T^{n-1}\) is the minimal polynomial, because
\(2\) has order \(n-1\) modulo \(n\). That relation has odd parity.
Thus evaluation is injective, and both spaces have dimension \(n-1\).

In particular the field elements \(\xi_i=Y_{a_i}(\mathbf t)\) are
binary independent. With \(\eta_i=L_i(\mathbf t)\), the equations derived from the diagonal and slope
checks are exactly \eqref{eq:descent-product}--\eqref{eq:descent-moore}.
Lemma~\ref{lem:frobenius-descent} gives a symmetric binary
matrix \(C\) with zero diagonal such that
\(\eta_i=\sum_j C_{ji}\xi_j\). Injectivity of evaluation now gives an
identity of linear forms, \(L_i=\sum_j C_{ji}Y_{a_j}\). Hence
\[
 P=\sum_{i,j} C_{ji}Y_{a_i}Y_{a_j}=0:
\]
the diagonal terms vanish and the other terms cancel in pairs.
The coefficients of \(P\) are the original edge and loop labels, so the
erased-column block has trivial kernel.
\end{proof}

\begin{proof}[Proof of the rank assertion in Theorem~\ref{thm:multislope}]
Apply the rank count in the proof of
Proposition~\ref{prop:primitive-slope-rank} with
\(\Lambda=\{2^k:1\le k\le\rho-2\}\), interpreted modulo \(n\).
Since \(2\) has order \(n-1\) and \(\rho-2\le(n-3)/2\), these slopes
are distinct, none is \(\pm1\), and no two are inverse to one another.
For \(n=3\) the only case is \(\rho=2\), with \(\Lambda\) empty; the same
count applies. Thus the rank is \(n+(\rho-1)(n-1)=\rho n-(\rho-1)\).
\end{proof}

\section{A One-Dimensional Extension Criterion and Finite Obstructions}
\label{app:one-bit}

We ask whether the published slope-two triple-node code can be enlarged by
one dimension while preserving triple-node correction.  The criterion and
finite obstructions below concern extensions of that fixed code.  The
optimal codes in Proposition~\ref{prop:finite-completions} use a different
parity-check construction.

Assume that \(n\geq5\) is prime and that \(2\) is primitive modulo \(n\), and
set
\[
 r_0\triangleq3n-2.
\]
Let
\[
 H_0\in\F_2^{r_0\times N}
\]
be a full-row-rank parity-check matrix for the published triple-node code
\(\cC_3=\cC_3(2)\) from \eqref{eq:literal-specialization},
with columns indexed by \(E_n\), and put
\(\cC_0\triangleq\ker H_0\).
Here a one-dimensional extension means a supercode
\[
 \cC_0\subseteq\cC',
 \qquad
 \dim\cC'=\dim\cC_0+1;
\]
equivalently, its parity-check space is a codimension-one subspace of the row
space of \(H_0\).  The extension is correcting if it still corrects every
triple of failed vertices.  For every failed triple
\(A\in\binom{[n]}3\), define
\[
 H_{0,A}\triangleq H_0|_{\St(A)}.
\]
This submatrix has
\(r_0-1\) columns and, because \(\cC_3\) corrects every triple of failed
vertices, rank \(r_0-1\).  Let
\[
 \nu_A\in\F_2^{r_0}\setminus\{0\}
\]
be its unique left-null vector, that is,
\[
 \nu_A^{\mathsf T}H_{0,A}=0.
\]
For \(z,x\in\F_2^{r_0}\), write
\[
 z\mathbin{\cdot}x\triangleq\sum_{i=1}^{r_0}z_ix_i
\]
for the standard binary coordinate inner product.

\begin{proposition}[One-dimensional extension criterion]
\label{prop:one-word}
There is a correcting one-dimensional extension of the published code if
and only if the affine system
\begin{equation}
 z\mathbin{\cdot}\nu_A=1
 \qquad
 \text{for every }A\in\binom{[n]}3,
 \qquad z\in\F_2^{r_0},
\label{eq:affine-extension}
\end{equation}
is consistent.  In particular, a set
\(\mathcal T\subseteq\binom{[n]}3\) of odd cardinality satisfying
\begin{equation}
 \bigoplus_{A\in\mathcal T}\nu_A=0
\label{eq:odd-certificate}
\end{equation}
certifies that no such extension exists.
\end{proposition}

\begin{proof}
The map \(H_0\) is onto and has kernel \(\cC_0\). Hence a one-dimensional
extension is \(H_0^{-1}(\{0,z\})\) for a nonzero syndrome
\(z\in\F_2^{r_0}\). On the coordinates erased by a triple \(A\), the map
\(H_{0,A}\) is injective, and its image is the hyperplane
\(\{w:\nu_A\mathbin{\cdot}w=0\}\). The extension adds a word supported there
precisely when \(z\) belongs to that hyperplane. Avoiding every such
word is therefore equivalent to \(z\mathbin{\cdot}\nu_A=1\) for every \(A\).
Summing these equations over any odd family with
\(\bigoplus_A\nu_A=0\) gives \(0=1\).
\end{proof}

Conversely, over \(\F_2\) the system \eqref{eq:affine-extension} is
inconsistent exactly when such an odd set exists.  Exact odd-dependency
certificates exist for the sample lengths
\[
 n=5,11,13,19;
\]
the complete certificates and an independent exact check of each are
archived in the verification package
\cite{ZabokritskiyYohananov2026Verification}.
For example, at \(n=13\), the eleven triples
\[
 \{0,1,k\},
 \qquad 2\le k\le12,
\]
already satisfy \eqref{eq:odd-certificate}.  Thus simply removing one
independent check from the row space of \(H_0\) cannot close the
one-bit gap at these lengths.

\section{Dual Weight Data and the Failure of a Univariate MacWilliams Transform}
\label{app:macwilliams}

By \cite[Thm.~1]{PinheiroMachadoFirer2019}, a combinatorial metric admits a
MacWilliams-type identity if and only if its irredundant cover is a partition
into equal-sized blocks.  For \(n\geq2\), the coordinate cover \(\{N_i:i\in[n]\}\) is irredundant,
because star \(i\) uniquely contains the loop \(\langle v_i,v_i\rangle\), and it
is not a partition, because distinct stars overlap on ordinary-edge coordinates.  Hence, for general linear codes, the univariate node-weight
enumerator does not determine the corresponding enumerator of
the dual.  The following two-vertex witness confirms this in our ambient
space and duality convention.

\paragraph{A two-vertex witness.}
In the binary space \(\cG_2(2)\), whose three coordinates are the two loops
and the edge, for \(0\leq i\leq j\leq1\), let
\(\mathbf u_{i,j}\in\cG_2(2)\) denote the graph word whose only nonzero label
is \(e_{i,j}=1\), and let
\[
 \mathcal D_1\triangleq\langle\mathbf u_{0,0}\rangle,
 \qquad
 \mathcal D_2\triangleq\langle\mathbf u_{0,1}\rangle.
\]
Both codes have node-weight enumerator
\[
 W_{\mathcal D_1}^{\rm N}(z)=W_{\mathcal D_2}^{\rm N}(z)=1+z.
\]
Their edge-coordinate duals are respectively specified by
\[
 y_{00}=0
 \qquad\text{and}\qquad
 y_{01}=0.
\]
The node-weight enumerators of these duals are
\[
 W_{\mathcal D_1^\perp}^{\rm N}(z)=1+3z,
 \qquad
 W_{\mathcal D_2^\perp}^{\rm N}(z)=1+2z+z^2.
\]
Thus equal primal enumerators can have unequal dual enumerators.

There is nevertheless an exact Fourier identity.  Fix a nontrivial
additive character \(\chi\) of \(\F_q\).  For every integer
\(0\leq w\leq n\), define the function
\[
 \mathcal K_w:\cG_q(n)\longrightarrow\mathbb C,
 \qquad
 \mathcal K_w(y)\triangleq
 \sum_{\substack{x\in\cG_q(n)\\\wtN(x)=w}}
 \chi\bigl(\langle x,y\rangle_E\bigr).
\]
Character orthogonality gives, for every linear code,
\begin{equation}
 A_w^{\rm N}(\cC)
 =
 \frac{1}{|\cC^\perp|}
 \sum_{y\in\cC^\perp}\mathcal K_w(y).
\label{eq:fourier-weight}
\end{equation}
The preceding witness shows that \(\wtN(y)\) alone does not determine
all the values \(\mathcal K_w(y)\). For each \(w\), the value depends
only on the support type of \(y\), the orbit of \(\supp(y)\) under vertex relabeling:
relabeling vertices preserves node weight and the edge-coordinate inner
product, and multiplying a coordinate \(y_e\) by a nonzero scalar can be
absorbed into \(x_e\), which leaves \(\wtN(x)\) unchanged.  Applied to
\(\cC^\perp\) in place of \(\cC\), identity \eqref{eq:fourier-weight}
therefore recovers the node-weight distribution of the dual from the
numbers of codewords of \(\cC\) of each support type.

For a partition of the nonempty subsets of \(E_n\), the class counts of a
linear code \(\cC\) are the numbers of codewords of \(\cC\) whose support belongs
to each class. Call the partition sufficient if, for every linear code \(\cC\),
the class counts of \(\cC\) determine the node-weight distribution of
\(\cC^\perp\).  Since \(\mathcal K_w(y)\) depends only on \(\supp(y)\), the
same use of \eqref{eq:fourier-weight} shows that grouping the supports by the
values \(\mathcal K_w\), \(0\leq w\leq n\), as in the dual partition of
\cite[Def.~2.1]{GluesingLuerssen2015}, gives a sufficient partition.  It is
the coarsest one: every sufficient partition refines it.  Indeed, for
\(y\ne0\) every nonzero word of \(\langle y\rangle\) has support
\(\supp(y)\), so \eqref{eq:fourier-weight} gives
\[
 A_w^{\rm N}\bigl(\langle y\rangle^\perp\bigr)
 =\frac1q\Bigl(A_w^{\rm N}(n,q)+(q-1)\,\mathcal K_w(y)\Bigr).
\]
If the supports of two words \(y,y'\) lie in one class of a sufficient
partition, then \(\langle y\rangle\) and \(\langle y'\rangle\) have the same
class counts, so their duals have the same node-weight distribution and the
display gives \(\mathcal K_w(y)=\mathcal K_w(y')\) for every \(w\).

\end{appendices}

\setlength{\bibsep}{0.7em}

\end{document}